\pdfoutput=1
\RequirePackage[T1]{fontenc}
\documentclass[12pt]{article}

\usepackage[height=8.85in,width=6.45in]{geometry}

\usepackage[utf8]{inputenc}
\usepackage{amsmath}
\usepackage{amssymb}
\usepackage{mathtools}
\numberwithin{equation}{section}
\usepackage{slashed}
\usepackage{braket}
\usepackage{enumitem}
\usepackage[svgnames]{xcolor}
\usepackage[colorlinks,citecolor=DarkGreen,linkcolor=FireBrick,urlcolor=FireBrick,linktocpage,unicode,psdextra]{hyperref}
\usepackage{cite}
\usepackage{graphicx}
\usepackage{tikz}
\usetikzlibrary{shapes.geometric}
\usetikzlibrary{calc}
\tikzstyle{circ}=[circle,draw,inner sep=1pt]
\tikzstyle{hexagon}=[draw, regular polygon,regular polygon sides=6,inner sep=2pt]
\tikzset{>=stealth}
\usepackage[bottom]{footmisc}

\usepackage{times}
\usepackage{courier}
\usepackage{bm}
\usepackage{subfig}

\usepackage{xcolor}
\usepackage{mdframed}
\newenvironment{claim}{  \begin{mdframed}[linecolor=black!0,backgroundcolor=black!10]\noindent\slshape\ignorespaces}{\end{mdframed}}

\renewenvironment{figure}[1][]{
  \begin{originalfigure}[#1]
    \begin{mdframed}[linecolor=black!0,backgroundcolor=black!1]
}{
    \end{mdframed}
  \end{originalfigure}
}

\usepackage{xcolor}

\def\Nequals#1{$\mathcal{N}{=}#1$}
\def\ind{\mathop{\mathrm{ind}}\nolimits}
\def\tr{\mathop{\mathrm{tr}}\nolimits}
\def\sT{\mathsf{T}}
\def\sY{\mathsf{Y}}
\def\sW{\mathsf{W}}

\def\cH{\mathcal{H}}
\def\bC{\mathbb{C}}
\def\bZ{\mathbb{Z}}
\def\bR{\mathbb{R}}
\def\mf{\mathop{\mathrm{mf}}\nolimits}
\def\MF{\mathrm{MF}}
\def\TMF{\mathrm{TMF}}
\def\tmf{\mathrm{tmf}}
\def\KO{\mathrm{KO}}
\def\KU{\mathrm{KU}}

\def\beq#1\eeq{\begin{align}#1\end{align}}

\usepackage[all]{xy}
\usepackage{amsthm}

\makeatletter
\@addtoreset{equation}{section}
\makeatother
\theoremstyle{remark}
\newtheorem{rem}[equation]{Remark}

\theoremstyle{plain}
\newtheorem{lem}[equation]{Lemma}
\newtheorem{prop}[equation]{Proposition}

\def\spin{\mathrm{Spin}}

\def\multi{\mathrm{mul}}
\def\id{\mathrm{id}}

\def\vphi{\varphi}
\def\sigm{\varphi}
\def\physZtwo{\bZ_2}

\def\PFirst{
Let $\sT$ be an SQFT with gravitational anomaly $n=8k+2$.
Assume that $\sT \times S^1$ is null-bordant in the space of supersymmetric field theories with gravitational anomaly $8k+3$.
Then $\varphi(\sT)$ is a mod-2 reduction of a weakly-holomorphic integral modular form of weight $4k+2$.
More explicitly, it is given by a linear combination over $\physZtwo$ of $c_4^{k-1-3j}c_6\Delta^j \equiv \Delta^j$,
where $j\le \lfloor \frac{k-1}{3}\rfloor$. 
}
\def\MFirst{
For $n = 8k+2$, 
let $A_n \subset \pi_n\TMF$ denote the kernel of the multiplication by $\eta \in \pi_1 S$. 
Restricted to this subgroup, 
the image of $\vphi|_{A_n}:A_n  \to \pi_n \KO((q)) \simeq \bZ/2((q))$ is 
contained in the mod-two reduction of $\MF_{4k+2}$. 
}
\def\MSecond{
The image of $\vphi:\pi_n\TMF\to \pi_n\KO((q))\simeq \bZ/2((q))$ when $n=8k+2$ is contained in the mod-two reduction of $\MF_{4k}$. 
}

\begin{document}

\begin{titlepage}

\begin{flushright}
\end{flushright}

\vskip 3cm

\begin{center}

{\Large \bfseries  Remarks on  mod-2 elliptic genus}

\vskip 2cm
Yuji Tachikawa$^1$, 
Mayuko Yamashita$^2$,
and Kazuya Yonekura$^3$
\vskip 1cm

\begin{tabular}{ll}
1 & Kavli Institute for the Physics and Mathematics of the Universe (WPI), \\
& University of Tokyo,   Kashiwa, Chiba 277-8583, Japan \\
2 & Research Institute for Mathematical Sciences, Kyoto University, \\
&Kitashirakawa-Oiwake-Cho, Sakyo-ku, Kyoto, 606-8502 Japan\\
3 & Department of Physics, Tohoku University,\\
&Sendai, Miyagi, 980-8578, Japan \\
\end{tabular}

\vskip 2cm

\end{center}

\noindent 
\textbf{For physicists:} 
For supersymmetric quantum mechanics, there are cases when a mod-2 Witten index can be defined, even when a more ordinary $\bZ$-valued Witten index vanishes.
Similarly, for 2d supersymmetric quantum field theories, there are cases when a mod-2 elliptic genus can be defined, even when a more ordinary elliptic genus 
vanishes. 
We study such mod-2 elliptic genera in the context of \Nequals{(0,1)} supersymmetry,
and show that they are characterized by mod-2 reductions of integral modular forms, under some assumptions.

\medskip

\noindent\textbf{For mathematicians:}
We study the image of the standard homomorphism \[
\pi_n\TMF\to \pi_n\KO((q))\simeq \bZ/2((q))
\] for $n=8k+1$ or $8k+2$, by relating them to the mod-2 reductions of integral modular forms.
\end{titlepage}

\setcounter{tocdepth}{2}
\tableofcontents


\section{Introduction}
\label{sec:introduction}

This paper is mostly written toward theoretical physicists, but also contains some rigorous mathematical results. 
We start with an introduction for physicists in Sec.~\ref{sec:intro_phys}.
A short introduction for mathematicians is given in Sec.~\ref{sec:intro_math}.

\subsection{Introduction for theoretical physicists}
\label{sec:intro_phys}

\paragraph{Witten index:} 
Let us consider  a supersymmetric quantum mechanical system.
Its Witten index \cite{Witten:1982df,Witten:1982im} 
is defined by \begin{equation}\label{eq:Witten}
I:=\tr_{\cH} (-1)^F e^{-\beta H},
\end{equation}
where $\cH$ is the Hilbert space of the system, $(-1)^F$ is the fermion parity,
$H$ is the Hamiltonian, and $\beta$ is the inverse temperature.
Using the fact that $H=Q^2$ where $Q$ is the supercharge satisfying $\{(-1)^F,Q\}=0$,
we can show that the Witten index $I$ is independent under any supersymmetry-preserving continuous deformations including the change in $\beta$.\footnote{\textcolor{black}{
This is due to the cancellations between states $\ket{+}$ and $\ket{-}$, where $\ket{+}$ is an eigenstate of $H$ and $(-1)^F$ with eigenvalues $H=E>0$ and $(-1)^F=+$, 
and $\ket{-}$ is defined by $\ket{-}=E^{-1/2}Q\ket{+}$. The precise cancellations hold only if the energy spectrum is discrete. See e.g. \cite{Dabholkar:2021lzt} and references therein for what happens when the spectrum is not discrete. The same remark applies to the elliptic genus defined below. }
}
As such, it provides us with an important $\bZ$-valued invariant in the study of such systems. 

\paragraph{Elliptic genus:} 
In supersymmetric quantum field theories,
the Witten index can often be further refined.
For example, for a two-dimensional (2d) theory with \Nequals{(0,1)} supersymmetry, we can consider its partition function on $T^2$ with periodic spin structure, which is \begin{equation}
Z=\tr_{\cH} (-1)^F q^{(H+P)/2} \bar q^{(H-P)/2},\label{foo}
\end{equation} where $H$ is the Hamiltonian, $P$ is the momentum along the spatial $S^1$,\footnote{%
Our normalization of $P$ is such that $e^{i\theta P}$ generates a rotation by $\theta$.
We then normalize $H$ so that the massless excitations satisfy $H=|P|$.
}
and $q=e^{2\pi i\tau}$ specifies the complex structure of the torus.
As $H-P=Q^2$ where $Q$ is the supercharge,
only the states with $H=P$ contribute, making $Z$ independent of $\bar \tau$.
This allows us to rewrite $I(\tau):=Z$ as
\begin{align}
I(\tau)  &= \tr_{\cH} (-1)^F q^P e^{-\beta Q^2} \nonumber \\
&=\sum_k q^k \tr_{\cH_k} (-1)^F e^{-\beta Q^2}. \label{bar0}
\end{align}
where $\cH_k$ are the eigenspaces of the momentum operator $P$ with eigenvalue $k$.

The expression  \eqref{bar0} makes it clear
that $I(\tau)$ is the generating function of the ordinary Witten index
at each eigenspace of $P$,
and the expression \eqref{foo} makes it clear that $Z=I(\tau)$ 
should be modular invariant up to a phase fixed by the gravitational anomaly of the theory.
This $I(\tau)$ is known as the elliptic genus of the 2d supersymmetric theory in the physics literature \cite{Witten:1986bf}.

\paragraph{Mod-2 Witten index:}
Coming back to the case of general supersymmetric systems,
there are cases where the ordinary Witten index $I$ vanishes, but where a subtler, $\physZtwo$-valued invariant can be extracted.\footnote{%
Throughout this paper, we use physicists' convention so that $\physZtwo$ is $\{0,1\}$,
and  therefore is \emph{not} the ring of $2$-adic integers. 
We also use the notation $\bZ/2=\{0,1\}$ more common in modern mathematical literature interchangeably. 
}
This is when the system is invariant under the time-reversal operator $T$
such that \begin{equation}
T^2=1, \quad
T(-1)^F=-(-1)^FT,\quad
Q(-1)^F=-(-1)^FQ,\quad
TQ=-QT.\label{bosh}
\end{equation}
We give a detailed explanation of this algebra in Appendix~\ref{sec:cliff}.

The irreducible representation of this algebra is two-dimensional when $Q^2=0$ with the structure
\begin{equation}
\ket + \stackrel{T}\longleftrightarrow \ket -
\end{equation} 
while it is four-dimensional when $Q^2=E>0$ with the structure \begin{equation}
\begin{array}{r@{}ccc@{}l}
&\ket + &\stackrel{T}\longleftrightarrow & \ket -& \\
\text{\footnotesize $Q$}& \updownarrow & & \updownarrow & \text{\footnotesize $Q$} \\
&\ket{-'} &\stackrel{T}\longleftrightarrow & \ket{+'}
\end{array}.
\end{equation}
Here, $\pm$ inside a ket denotes the eigenvalue of $(-1)^F$.
The two states $\ket{-}$ and $\ket{-'}$ cannot be the same by the following reason. 
If they were the same, we would have $TQ \ket{+} = c \ket{+}$ for a constant $c$.
However, the anti-linearity of $T$ implies that $  (TQ)^2 \ket{+} = |c|^2 \ket{+}$, while the above algebra implies that
$(TQ)^2\ket{+}=-Q^2\ket{+}=-E\ket{+}$ and $-E<0$, a contradiction. 

The structure of irreducible representations makes it clear that the ordinary $\bZ$-valued Witten index vanishes.
It also implies that $\dim \cH_{E=0}$, which is always even, 
change only by a multiple of four under deformations preserving supersymmetry.
Therefore, $\dim \cH_{E=0} /2 $ modulo 2 is a $\physZtwo$-valued invariant of the system,
which we call the mod-2 Witten index.
The supersymmetric quantum mechanics on a spin manifold $M$ of dimension $2$ mod $8$
has this structure, and the mod-2 Witten index of this system 
equals the mod-2 index of the Dirac operator on $M$, almost by definition.\footnote{%
For a different derivation of the mod-2 index, see e.g.~\cite[Sec.2.3.1]{Gaiotto:2019gef}. We discuss the details in Appendix~\ref{sec:cliff}. 
}

\paragraph{Mod-2 elliptic genus:}
Let us now come back to the case of 2d supersymmetric theories.
The CPT transformation\footnote{
The usual CPT theorem (or more precisely the CRT theorem where R is a reflection of a single space coordinate) 
is proved in flat Minkowski spacetimes $\bR^{d-1,1}$ by using the fact that the Wick rotated 
Lorentz group $\mathrm{SO}(d)$ (or its spin cover) is connected~\cite{Streater:1989vi}. 
When we compactify the space, the full Lorentz symmetry $\mathrm{SO}(d)$ is not available and the situation becomes subtle. 
When there is no gravitational anomaly,
we may assume axioms of functorial QFT and the CRT transformation may be constructed within the axioms.
(For instance, see \cite[Example~2.24]{Yonekura:2018ufj} for the construction of time-reversal symmetry in topological QFT which
can also be modified to construct CRT.)
When gravitational anomalies exist, unmodified axioms of functorial QFT fail and the CRT will suffer from anomalies.
}
provides a time-reversal operator $T$ acting on the eigenspace $\cH_{k}$ of the momentum operator $P$ appearing in the expression \eqref{bar0}.

The algebra formed by $T$ and $(-1)^F$ is known to be determined by the gravitational anomaly of the 2d theory, which is specified by its anomaly polynomial $n p_1/48$,
where $p_1$ is the first Pontryagin class and $n\in\bZ$ equals $2(c_R-c_L)$ when the theory is conformal in our normalization.
When $n\equiv 2$ mod 8, the algebra is given by \eqref{bosh} as explained in Appendix~\ref{sec:cliff} (see also \cite{Delmastro:2021xox}).
As such, the ordinary Witten index of each $\cH_k$ vanishes, making the ordinary elliptic genus to vanish.
We can still form the generating function \begin{equation}
I_2(\tau):=\sum_k q^k  \frac{\dim{\cH_{k,E=0}}}2 \qquad ( n \equiv 2 \bmod 8),\label{qoo}
\end{equation} of the mod-2 Witten index of each $\cH_k$,
where the coefficients are now valued in $\physZtwo$.
We can call this invariant the mod-2 elliptic genus of the theory.

We can extend the definition of the mod-2 elliptic genus to the case $n\equiv 1$ mod 8.
When $n$ is odd, $(-1)^F$ is ill-defined (see Appendix~\ref{sec:cliff} for the precise meaning of this statement) and the algebra is
\beq
T^2=1, \quad TQ=-QT.
\eeq
We can define the mod-2 elliptic genus by the formula\footnote{
When $n$ is odd, there is a known subtlety in the meaning of $\dim{\cH}$.  
The definition of $\dim{\cH} \in \mathbb{Z}$ used in \eqref{qoo2} is such that the path integral computation for $\tr_{\cH}  e^{-\beta H}$
gives the ``dimension'' $\sqrt{2} \dim{\cH}$.
}
\beq
I_2(\tau):=\sum_k q^k  \dim{\cH_{k,E=0}} \qquad ( n \equiv 1 \bmod 8).\label{qoo2}
\eeq

One major difference of the mod-2 elliptic genus from the ordinary elliptic genus is that the modular property of \eqref{qoo} and \eqref{qoo2} is no longer clear,
due to two closely related reasons. 
Firstly, we usually understand the modular property of a modular function in the context of complex analysis, and the formal power series with $\physZtwo$ coefficients do not fit well in this viewpoint.\footnote{%
It is unclear to the authors whether the theory of $p$-adic modular forms is relevant here or not.
}
Secondly, the interpretation of $I(\tau)$ as the partition function of the theory on $T^2$ with a complex structure $\tau$ was crucial as the physical reason why it had a modular property.
The expression \eqref{qoo} and \eqref{qoo2} of $I_2(\tau)$ no longer has that interpretation, 
preventing us from gaining an understanding from this approach.

The aim of this paper is then to show, by indirect means and under some assumptions,
that $I_2(\tau)$ is still, roughly-speaking, given by a mod-2 reduction of a modular function with integral $q$-expansion coefficients,
in a sense which we make precise below.

\paragraph{Remarks on notations and conventions:}
In this paper, we mean by SQFT 2d supersymmetric quantum field theories with \Nequals{(0,1)} supersymmetry, unless otherwise explicitly mentioned.

As remarked in a previous footnote, we use $\bZ/2 =\bZ_2$.
We will sometimes use generalized cohomology theories like $\KO$-theory, especially in the mathematical part of the paper. 
For a generalized cohomology theory or spectrum $E$,
we have $\pi_n E=E_n(pt)=E^{-n}(pt)$ where $\pi_n$ is the $n$-th homotopy group of the spectrum, $pt$ is a point, and $E_\bullet$ and $E^\bullet$
are homology and cohomology theories associated to $E$, respectively. 
For an abelian group $\mathbb{A}$, the notations $\mathbb{A} [x]$, 
$\mathbb{A}[[x]]$ and $\mathbb{A}((x))$
mean the abelian groups of polynomials, formal power series, and Laurent series of variable $x$ with $\mathbb{A}$ coefficients, respectively.

We use the identifications for $\KO$-theory given by
$\KO_{1+8k}(pt) \simeq  \bZ/2$, $\KO_{2+8k}(pt) \simeq \bZ/2$, $\KO_{8k}(pt) \simeq \bZ$ and $\KO_{4+8k}(pt) \simeq \bZ$.
The last isomorphism $\KO_{4+8k}(pt) \simeq \bZ$ requires care, however. For instance, the Witten index defined as in \eqref{eq:Witten}
for $n \equiv 4 \bmod 8$ is always a multiple of 2 by Kramers degeneracy as explained in Appendix~\ref{sec:cliff},
and hence takes values in $2\bZ$. Then we need to divide the index by 2 for the identification $\KO_{4+8k}(pt) \simeq \bZ$.
Therefore, we will use the convention that the index for $n \equiv 4 \bmod 8$ is divided by 2.
From now on we define
\beq
I(\tau) &= \kappa \tr_{\cH} (-1)^F q^{(H+P)/2} \bar q^{(H-P)/2} = \kappa \sum_k q^k \tr_{\cH_k} (-1)^F e^{-\beta Q^2}, \nonumber \\
 \kappa &= \left\{ \begin{array}{cc} 1 & (n \equiv 0 \bmod 8), \\   1/2 & (n \equiv 4 \bmod 8). \end{array} \right. \label{bar}
\eeq

\paragraph{Organization of the paper:} 
The rest of the paper is organized as follows.
In Sec.~\ref{statements}, 
we make precise statements of two results we obtain in this paper.
To motivate these statements,
we study the mod-2 elliptic genus of \Nequals{(0,1)} supersymmetric sigma models.

In Sec.~\ref{yonekura},
we give physical and mathematical derivations of the first result.
Namely, the mod-2 elliptic genus of a theory with gravitational anomaly $n=8k+2$  is given in terms of the mod-2 reduction of an integral modular form of weight $4k+2$,
under the assumption that the theory times the \Nequals{(0,1)} sigma model on $S^1$ with periodic spin structure is `null-bordant in the space of \Nequals{(0,1)} supersymmetric quantum field theories'; this concept was introduced 
in \cite{Gaiotto:2019asa,Gaiotto:2019gef,Johnson-Freyd:2020itv} and utilized in \cite{Yonekura:2022reu}.

In Sec.~\ref{yamashita},
we give a mathematical derivation of the second result.
Namely, we provide a characterization of the image of the homomorphism $\TMF_n(pt) \to \KO_n((q))(pt) \simeq \bZ/2((q))$
when $n=8k+1$ or $n=8k+2$.
This turns out to be an easy corollary of the detailed description of $\TMF_n(pt)$ in \cite{BrunerRognes}.
Translating this mathematical statement into the information on the mod-2 elliptic genus,
assuming the validity of the Stolz-Teichner conjecture \cite{StolzTeichner1,StolzTeichner2},
we find that the mod-2 elliptic genus of a theory with gravitational anomaly $n=8k+1$ or $8k+2$ is given by a mod-2 reduction of an integral modular form of weight $4k$.
Note the difference in the weight of the modular forms, $4k$ vs.~$4k+2$, between
the first result and the second.
We check the compatibility, again by using the detailed data provided in \cite{BrunerRognes}.

\paragraph{Previous works and further directions:}
Before proceeding, let us briefly discuss previous works and further directions.
The authors consider the mod-2 elliptic genus a natural combination of the mod-2 index and the ordinary elliptic genus,
both in the context of mathematics and physics.
Because of this, it was somewhat surprising to the authors that they could only find two existing references, namely \cite{Ochanine1991,Liu1992}, on this subject.
Both are papers in pure mathematics,
where the mod-2 version of the Ochanine genus (i.e.~the elliptic genus of \Nequals{(1,1)} supersymmetric sigma models on spin manifolds) was studied.

As for future directions, we should mention that we considered only the most basic case
of the mod-2 version of the supersymmetric indices in two dimensions.
For example, in the presence of a global symmetry group $G$, 
mod-2 indices can be generalized by counting the number of specific irreducible representations of $G$;
see \eqref{eq:KOG} in Appendix~\ref{sec:cliff} when $G$ is anomaly free after reduction to 1d.
Also, in principle, higher-dimensional superconformal indices of $d$-dimensional theories 
on manifolds of the form $S^1 \times M_{d-1}$ for a $(d-1)$-manifold $M_{d-1}$
can have mod-2 versions. It would be interesting to explore these issues further.

\subsection{Introduction for mathematicians}\label{sec:intro_math}

One of the important results in the theory of topological modular forms is the determination,
originally announced in \cite{Hopkins2002} and detailed e.g.~in\cite{BrunerRognes},
of the image of the standard homomorphism
\begin{equation}
\pi_n\TMF \to \MF_{n/2}
\end{equation}
where
\begin{equation}
\MF_\bullet=\bZ[c_4,c_6,\Delta,\Delta^{-1}]/(c_4^3-c_6^2-1728\Delta)
\end{equation}
is the ring of weakly-holomorphic integral modular forms,
which we identify with its image in $\bZ((q))$ after the $q$-expansion.
This map is non-trivial only when $n=8k$ or $8k+4$,
and is known to agree with the map
\begin{equation}
\varphi:\pi_n\TMF \to \pi_n\KO((q)) \simeq \bZ((q))
\end{equation} 
which is the homomorphism on the homotopy groups induced by 
the morphism at the level of spectra, \begin{equation}
\varphi: \TMF\to \KO((q)),
\end{equation} whose construction is detailed e.g.~in \cite[Appendix A]{HillLawson}.

The aim of this paper is to study the analogous question when $n=8k+1$ or $n=8k+2$.
Namely, we characterize the image of the homomorphism \begin{equation}
\vphi:\pi_n\TMF \to \pi_n\KO((q)) \simeq \bZ/2((q))
\end{equation}
when $n=8k+1$ or $8k+2$,
in terms of mod-2 reductions of integral modular forms.

Below, we use the standard convention that $\eta$ denotes the generator of $\pi_1(S)=\bZ/2$,
i.e.~the class of $S^1$ with the standard framing.
We also follow the standard abuse of notation that $\eta$ also denotes its image in various spectra we use, $\pi_1 M\spin$, $\pi_1M\mathrm{String}$, $\pi_1\TMF$, and $\pi_1\KO$.
In particular, $\eta$ and $\eta^2$ generate $\pi_1\KO\simeq \bZ/2$ and $\pi_2\KO\simeq \bZ/2$ respectively.
It also turns out $\vphi(\eta)\in \pi_1\KO((q))$ agrees with $\eta\in\pi_1\KO \subset \pi_1\KO((q))$.

\textcolor{black}{We phrase our results in terms of the mod-2 reductions of modular forms. Namely,
we consider $\MF_n\subset \bZ((q))$ using the $q$-expansion, and call the image under \begin{equation}
\MF_n \subset \bZ((q)) \xrightarrow{\mod 2} \bZ/2((q)) 
\end{equation} the mod-two reductions of $\MF_n$.
The right hand side can then be identified canonically with $\pi_{8k+2}\KO((q))$.
This terminology follows that used in \cite{Ochanine1991,Liu1992}.
As $c_4\equiv c_6\equiv 1$ mod 2, the mod-two reductions of $\MF_n$ can all be written in terms of mod-two reductions of $j=(c_4)^3/\Delta$.
A short computation reveals that \begin{equation}
\text{mod-two reductions of $\MF_n$}  = 
\begin{cases}
(\bZ/2) j^{-\lfloor \frac k3\rfloor} [j] & \text{when $n=4k$,} \\
(\bZ/2) j^{-\lfloor \frac{k-1}3\rfloor} [j] & \text{when $n=4k+2$.} 
\end{cases}
\label{j}
\end{equation}
}
Two statements we prove are the following:
\begin{claim}
\textbf{First statement:}\\
\MFirst
\end{claim}

\begin{claim}
\textbf{Second statement:}\\
\MSecond
\end{claim}

Here, we concentrated on the case $n=8k+2$, since the case $n=8k+1$ can be easily reduced to the case $n=8k+2$ by multiplying by $\eta$.
The proof of the first statement is given in Sec.~\ref{subsec_proof_first_math},
and that of the second statement can be found in Sec.~\ref{yamashita}.
Note the difference in the degree of integral modular forms, $4k+2$ vs.~$4k$,
between the first and the second statement.
\textcolor{black}{This leads to the difference in the possible negative powers of $j$ appearing in the results,}
and this difference allows us to detect some elements in $\pi_n\TMF\setminus A_n$.
\textcolor{black}{For example, $\eta^2\in \pi_2\TMF$ is sent to $\eta^2 \in \pi_2\KO((q))$ which is $1\in \bZ/2((q))$.
Now, the mod-2 reduction of $\MF_0$ is $\bZ/2[j]$, 
whereas the mod-2 reduction of $\MF_2$ is $(\bZ/2) j[j]$.
$1$ is clearly in the former while it is not in the latter.
And indeed, $\eta\eta^2\neq 0$ in $\pi_3\TMF$.}

\textcolor{black}{We remark that, after $2$-adic completion, the morphism $\psi \colon \TMF \to \KO((q))$ is well-known to factor through the $K(1)$-localization of $\TMF$,
whose homotopy groups satisfy $\pi_n L_{K(1)}\TMF=\pi_n\KO[j^{-1},j]$ \cite{LauresK(1)}. 
This fact, however, does not directly imply our results, since the $K(1)$-localization inverts $c_4$, whereas the information without inverting $c_4$ is essential for our statements,
which constrain the maximal power of $j^{-1}$ appearing in the image.}

Our proof of the first statement is homotopy theoretic and  
can be considered as the $\TMF$ version 
of a proposition on $\KO$ by Stong \cite[Proposition on p.343]{StongTextbook};
\textcolor{black}{we also note that the results of \cite{Ochanine1991,Liu1992} were also the Ochanine-genus version of the same proposition of Stong.}
The proof of the second statement, at present, relies on the detailed structure of $\pi_\bullet\TMF$ given in \cite{BrunerRognes} and is computational.
It would be desirable to have a more conceptual proof of the second statement.
We will also provide an outline for such a proof there, with the still missing parts clearly demarcated.
Before proceeding, we note that in this paper we also use $\bZ_2$ to mean $\bZ/2$, i.e.~the ring of integers modulo 2, not in the sense of the ring of 2-adic integers.

\section{Preliminaries}
\label{statements}

\subsection{Preparations}
We need to start with some preparations. 
Let us come back to the ordinary elliptic genus $I_\sT(\tau)$, defined in \eqref{bar} for the theory $\sT$ under consideration.
When the gravitational anomaly of $\sT$ is $n\in \bZ$, it has the transformation law 
\textcolor{black}{\begin{equation}
I_\sT(\tau+1) = e^{-2\pi i n/24} I_\sT(\tau),\qquad I_\sT(-\frac1{\tau})= e^{-2\pi i n/8}I_\sT(\tau)\label{mod-trans}
\end{equation}} under the generators of the modular transformations.
In particular, the momentum eigenvalues $k$ appearing in \eqref{bar} satisfy $k\in -\frac n{24} +\bZ$.\footnote{%
\textcolor{black}{
The phases in \eqref{mod-trans} may be found as follows;
see also the discussions in \cite[Sec.~3.5]{Seiberg:2018ntt}.
The transformation under $T:\tau\mapsto \tau+1$ simply comes from the non-integral part of the eigenvalues of $P$.
As for the transformation under $S:\tau\mapsto -1/\tau$, first recall that gravitational anomalies only affect the phase of partition functions (or correlation functions), 
and hence there exists a phase $\alpha(\tau) \in {\rm U}(1)$ such that $I_\sT(-1/\tau) = \alpha(\tau) I_\sT(\tau)$. Moreover, $\alpha(\tau)=I_\sT(-1/\tau)/I_\sT(\tau)$ is a holomorphic function of $\tau$ (at least for those $\tau$ for which $I_\sT(\tau) \neq 0$), and such a holomorphic ${\rm U}(1)$ phase must be a constant, $\alpha(\tau) = \alpha$. Let $\beta= e^{-2\pi i n/24}$ be the phase from $T$. 
We regard $S$ and $T$ as diffeomorphisms acting on the spin manifold $T^2$, and they generate the metaplectic group, the double cover of $SL(2,\mathbb{Z})$. 
From the relation $(ST)^3=S^4$, we conclude that $\alpha = \beta^3=e^{-2\pi i n/8}$. This is also consistent with the relation $S^8=1$ which implies $\alpha^8=1$. (We remark that we do not have $S^4=1$ since $S^4$ acts nontrivially on fibers of the spin bundle on $T^2$. We also remark that the above discussions, after some modification, are also valid for correlation functions and hence meaningful even if the partition function is zero, such as the case $n \neq 0 \mod 4$.)
The authors thank A. Ishige for pointing out an error in this equation in a previous version of the paper, leading to the clarification mentioned in this footnote.}
}
As $I_\sT(ia)$ for positive real $a$ should be real, we see that $I_\sT(\tau)$ vanishes unless $n\equiv 0$ or $4$ mod 8.

It is useful to consider, then, the combination \begin{equation}
\varphi_\sT(\tau) := \eta(\tau)^n I_\sT(\tau)
\end{equation} so that it has the transformation law \begin{equation}
\varphi_\sT(\tau+1)=\varphi_\sT(\tau),\qquad \varphi_\sT(-\frac1\tau)=\tau^{n/2} \varphi_\sT(\tau),
\end{equation}
where $\eta(\tau)=q^{1/24}\prod_{n>0}(1-q^n)$ is the Dedekind eta function,
which is known to behave as $\eta(\tau+1)=e^{2\pi i/24} \eta(\tau)$ and $\eta(-1/\tau)=(-i\tau)^{1/2} \eta(\tau)$.
This means that $\varphi_\sT(\tau)$ is a weakly-holomorphic modular form of weight $n/2$.\footnote{``Weakly holomorphic'' means that it is holomorphic 
except at $q=0$ where poles are allowed. }
From now on, we suppress the dependence on $\tau$ or $q$ from the notation 
and write  $\varphi(\sT)$ for $\varphi_\sT(\tau)$, namely
\beq
\varphi(\sT) := \varphi_\sT(\tau).
\eeq
The expression \eqref{bar} of $I(\tau)$ then says that the $q$-expansion coefficients of $\varphi(\sT)$ are all integers,
i.e.~$\varphi(\sT)$ is an integral weakly-holomorphic modular form of weight $n/2$.
Let us call $\varphi(\sT)$ as the Witten genus of the theory,
following the usage of mathematicians.

Here we recall that the ring $\mf$ of integral modular forms  is given by
\begin{equation}
\mf_{\bullet} = \bZ[c_4,c_6,\Delta]/(c_4^3-c_6^2-1728\Delta),
\end{equation}
where\begin{equation}
c_4=1+240\sum_{n>0} \sigma_3(n) q^n,\quad
c_6=1-504\sum_{n>0}\sigma_5(n)q^n,\quad
\Delta=\eta(\tau)^{24},
\end{equation}
are  the normalized Eisenstein series and the modular discriminant, respectively,
related by $c_4^3-c_6^2 = 1728\Delta$.
Here $\sigma_k(n):=\sum_{d|n} d^k$ is the divisor sum function,\footnote{
In other words, $\sum_{ n\ge 1} \sigma_k(n) q^n = \sum_{\ell \geq 1} \ell^k \frac{q^\ell}{1-q^\ell}$.
}
and the weights of $c_4$, $c_6$ and $\Delta$ are $4,6,12$, respectively.
The ring of weakly-integral modular forms $\MF$ is then obtained by inverting $\Delta$:
\begin{equation}
\MF_\bullet=\mf_\bullet[\Delta^{-1}].
\end{equation}
Our discussion so far can then be summarized by saying that the Witten genus $\varphi(\sT)$ of a theory whose gravitational anomaly is $n\in \bZ$ is an element of $\MF_{n/2}$.
Clearly, this is nonzero only when $n$ is a multiple of four, 
since all the generators of $\MF_\bullet$ have even weights.

It turns out to be convenient, also for the mod-2 elliptic genus \eqref{qoo} and \eqref{qoo2}, to introduce \begin{equation}
\varphi(\sT):=\eta(\tau)^n I_{2, \sT}(\tau),
\end{equation} which now takes values in $\physZtwo((q))$.
Here and below, we use the same symbol $\varphi(\sT)$ both for the $\bZ$-valued Witten genus and $\physZtwo$-valued Witten genus, 
depending on the gravitational anomaly $n\in \bZ$.
It is $\bZ$-valued when $n\equiv 0$ or $4$,
and $\bZ_2$ valued when $n\equiv 1 $ or $2$ mod 8.

Our aim is to characterize this mod-2 Witten genus $\varphi(\sT)$,
in terms of mod-2 reductions of integral modular forms.
For this purpose, it is useful to remember
\begin{equation}
c_4 \equiv c_6 \equiv 1, \qquad
\Delta\equiv \sum_{n\ge 1}  q^{(2n-1)^2} \mod 2.
\end{equation}

\subsection{Elliptic genera of supersymmetric sigma models}

Before proceeding, let us discuss the Witten genera $\varphi(M)$ 
for the \Nequals{(0,1)} supersymmetric sigma model on an $n$-dimensional manifold $M$.
This model has a scalar field taking values in $M$,
together with its right-moving fermionic superpartner, taking values in the tangent bundle $TM$.
Therefore it has the gravitational anomaly characterized by $n\in \bZ$.
The cancellation of the sigma model anomaly requires that $M$ is equipped 
not only with a spin structure,
but also with a choice of the $B$-field such that the equation $dH=p_1(TM)/2$ 
is solved at the integral level,
where $H$ is the gauge-invariant curvature of the $B$-field.
The necessity of this choice was pointed out in \cite{Witten:1985mj} and more systematically studied
in \cite{Witten:1999eg,Yonekura:2022reu}.
Let us call the entirety of this data as the differential string structure of the manifold,
and a manifold equipped with such a data simply as a string manifold.

The Witten genera can be computed in the ultraviolet using semi-classical methods \cite{Witten:1986bf}.
We simply obtain {
\begin{align}
\varphi(M)  &=  \eta(\tau)^n\ind\left( q^{-n/24}  \bigotimes_{m \geq 1} \left( \bigoplus_{k\ge 0} q^{mk}  \mathrm{Sym}^k TM \right) \right) \nonumber \\
&=  \ind\left(  \bigotimes_{m \geq 1} \left( \bigoplus_{k\ge 0} q^{mk}  \mathrm{Sym}^k (TM \ominus \bR^n) \right) \right)  \label{too}
\end{align} 
}%
where $\ind(V)$ 
for a real bundle $V$  are the $\bZ$- or $\physZtwo$-valued 
indices of Dirac operators on the spin bundle tensored with $V$.\footnote{ $V \ominus W$ for vector bundles $V$ and $W$ means a virtual vector bundle which we may also denote as $V - W$. 
The $k$-th symmetric power $\mathrm{Sym}^k$ acts as
$\mathrm{Sym}^k(V \ominus W) = \bigoplus_{\ell=0}^k \ominus^{\ell} \mathrm{Sym}^{k-\ell}(V) \otimes \wedge^{\ell} (W)$. }
Namely, $\ind$ is $\bZ$-valued when $n\equiv 0$ or $4$ mod 8,
and is $\bZ_2$ valued when $n\equiv 1$ or $2$ mod 8.
Furthermore, the definition of $\ind$ includes the factor $\kappa$ as in \eqref{bar}, that is, we divide the usual index by $2$ when $n = 4 \bmod 8$.
These expressions can be used to make sense of $\varphi(M)$ 
of spin manifolds as elements of $\bZ[[q]]$ (when $n\equiv 0, 4$ mod 8)
or  $\physZtwo[[q]]$ (when $n\equiv 1, 2$ mod 8).

Using the index theorem, the  expression \eqref{too} when $n\equiv 0, 4$ mod 8 can  be rewritten as \begin{equation}
\varphi(M)= \kappa \int_M \hat A(TM)\prod_{k=1}^\infty \prod_{j=1}^{n}
\frac{(1-q^k)}{ (1-q^k e^{u_j} ) }, 
\end{equation} where $ u_{1,\ldots,n}$ are the Chern roots of $TM_\bC$, and $\kappa$ is as in \eqref{bar}.
In \cite{Zagier1986}, this was further rewritten as
\begin{equation}
\varphi(M) =\kappa \int_M  \prod_{k>0}\exp\left(- \frac{2}{(2k)!}\frac{B_{2k}}{4k} E_{2k}(q)  \left( \frac12 \sum_{j=1}^{n} u_j^{2k} \right) \right)\label{dog}
\end{equation} 
where \begin{equation}
E_{2k}(q)= 1- \frac{4k}{B_{2k}} \sum_{d>0} \sigma_{2k-1}(d) q^d
\end{equation} is the Eisenstein series.
Here $B_{2k}$ is the Bernoulli number which is known to satisfy $4k/B_{2k}\in \bZ$;
this means in particular that $E_{2k}(q)\in \bZ[[q]]$.\footnote{%
We denoted $E_{4,6}$ also by $c_{4,6}$ in other parts of this paper,
following the usage in the literature on $\TMF$.}
As $E_{2k}$ for $2k\ge 4$ are modular forms of weight $2k$ while $E_2$ is only quasi-modular,
we see that the expression \eqref{dog} is modular when $p_1(TM)$ vanishes rationally
and quasimodular otherwise.
Combining with the fact that it has  integral $q$-expansion coefficients, we can summarize our discussions so far as follows:
\begin{claim}
$\varphi(M)$ of a string manifold $M$ of dimension $n\equiv 0$ or $4$ mod $8$ is in $\mf_{n/2}$, i.e.~an integral modular form of weight $n/2$.
$\varphi(M)$ of a spin manifold $M$ of dimension $n\equiv 0$ or $4$ mod $8$ is in $\bC[E_2,E_4,E_6]\cap\bZ[[q]]$, i.e.~the ring of integral quasimodular forms, again of weight $n/2$.
\end{claim} 

\subsection{Mod-2 elliptic genera of supersymmetric sigma models}\label{subsec_preliminary_mod2elliptic}
Using these preparations, 
let us next study the mod-2 elliptic genera of supersymmetric sigma models
as motivating examples.
The most basic cases are  $S^1$ and $T^2$ with periodic (i.e.~non-bounding) spin structures.
The expression 
\eqref{too} allows us to compute the Witten genera explicitly;
we simply find \begin{equation}
\varphi(S^1) =1, \qquad \varphi(T^2)=1.
\end{equation}
Here and below, we suppress the chosen spin structures for $S^1$ and $T^2$, as we always use the periodic spin structure.

As ordinary and mod-2 Witten genera are multiplicative, for any string manifolds $M_n$ of dimension $n \equiv 0 \bmod 8$, we  have \begin{equation}
\varphi(S^1\times M_n) = \varphi(T^2 \times M_n) = \varphi(M_n) \mod 2.
\end{equation} 
They are therefore given by the mod-2 reduction of integral modular forms of weight $n/2$.

Let us now study the mod-2 Witten genera of supersymmetric sigma models on general string manifolds $M_d$ of dimension $d=8k+1$ or $8k+2$. 
In fact, from the fact that $\varphi(S^1)=1$, one can check that $\varphi(S^1 \times M_{8k+1}) = \varphi(M_{8k+1})$.
Therefore, it suffices for us to consider only the case $d=8k+2$.

Here we follow the discussions of \cite{Ochanine1991,Liu1992} 
where the mod-2 version of the Ochanine index (i.e.~the mod-2 elliptic genus
of the \Nequals{(1,1)} supersymmetric sigma models) was considered.
There, the crucial step was the use of  \cite[Proposition on p.343]{StongTextbook}.

To explain the proposition, 
let us take a closed spin manifold $M_{8k+2}$.
Then $S^1\times M_{8k+2}$ is spin null-bordant\footnote{%
\label{foot:spin-null}%
This is due to the following. 
Any spin bordism class can be detected by its $\KO$ Pontrjagin numbers or Stiefel-Whitney numbers. 
The $\KO$ Pontjagin numbers vanish since $\pi_{8k+3}\KO=0$,
and the Stiefel-Whitney numbers vanish since it is a product with $S^1$.
}, 
i.e.~there is a compact spin manifold with a boundary $N_{8k+4}$ such that $\partial N_{8k+4}=S^1\times M_{8k+2}$.
A spin null bordism of $S^1\sqcup S^1$ is given by $S^1\times[0,1]$.
We can then paste $S^1 \times [0,1] \times M_{8k+2}$ with two copies of $N_{8k+4}$ along a common boundary to form a closed spin manifold $W_{8k+4}$.
With this setup,  the Proposition\footnote{%
\label{foot:analytic_stong}
The proof given in \cite{StongTextbook} uses methods of algebraic topology,
using the fact that the construction of $W_{8k+4}$ from $M_{8k+2}$ is that of the Toda bracket $\langle 2,S^1,M_{8k+2}\rangle$.
A rough outline of an analytic proof is as follows. 
By using flat metric on $S^1\times [0,1]$ and imposing an APS boundary condition, 
we can explicitly show $\ind(V \to M ) = \ind(V\to S^1\times [0,1]\times M)$ mod 2.
As the indices with appropriate APS boundary conditions are additive, we have 
$\ind(V\to W)=2\ind(V\to N)+\ind(V\to S^1\times [0,1]\times M)$.
Therefore we have $\ind(V\to M) =  \ind(V\to W)$ mod 2.
}
 states that
\begin{claim}
We have 
\begin{equation}
\ind(V \to M_{8k+2})
= \ind(V \to W_{8k+4}) \mod2,
\end{equation} 
where $V$ is any power of the tangent bundle $TW_{8k+4}$ that is also regarded on $M_{8k+2}$ as $TM_{8k+2} \oplus \bR^2$,
the left hand side is the mod-2 index on $M_{8k+2}$, and the right hand side is the $\bZ$-valued index on $W_{8k+4}$.
Here we use the definition of $\ind$ that includes the factor $\kappa=1/2$ for $n=8k+4$.
\end{claim}

This proposition immediately implies that \begin{equation}
\varphi(M_{8k+2}) = \varphi(W_{8k+4}) \mod 2
\end{equation}
in view of \eqref{too}. 
As an example, take $M_{8k+2}=T^2$. 
A spin null bordism of $T^3=S^1\times T^2$ is given by the half-K3 surface.\footnote{%
\label{foot:halfK3}
One realizes a K3 surface as an elliptic fibration over $\mathbb{CP}^1$ with 24 singularities.
We can group these singularities into two group of 12 singularities.
Then the $S^1$ surrounding a group of 12 has a non-bounding (i.e.~periodic) spin structure,
and the K3 can be cut into two. This provides a spin null bordism of $T^3$.}
Then $W_{8k+4}$ is the K3 itself. 
The formula \eqref{dog} says that \begin{equation}
\varphi(W_4) =- \int_{W_4} \frac{p_1}{48} E_2(q) 
\end{equation} and therefore \begin{equation}
\varphi(K3)=E_2(q).
\end{equation} We indeed find \begin{equation}
\varphi(T^2) = 1 = E_2(q)  \mod 2,
\end{equation}since $4/B_2=24$ is even.

Note that we can take $W_{8k+4}$ to be a string manifold 
if $N_{8k+4}$ can be taken to be a string manifold,
i.e.~if $S^1\times M_{8k+2}$ is null-bordant as a string manifold.
Combining with the statements at the end of the last subsection, we therefore conclude that:
\begin{claim}{}
$\varphi(M_{8k+2})$ is in general 
a mod-2 reduction of an element of $\bC[E_2,E_4,E_6]\cap \bZ[[q]]$, the ring of integral quasimodular forms, of weight $4k+2$.
Furthermore, if $S^1\times M_{8k+2}$ is null-bordant as a string manifold,
$\varphi(M_{8k+2})$ is a mod-2 reduction of an integral modular form of weight $4k+2$.
\end{claim}
Note that $\varphi(T^2)=1$ is \emph{not} a mod-2 reduction of an integral modular form of weight 2, since no modular form of weight 2 exists to start with.
This means that $T^3=S^1\times T^2$ is not null-bordant as string manifolds,
i.e.~there is no spin manifold $N_4$ with a $B$-field solving $dH=p_1/2$ such that $\partial N_4=T^3$.

\subsection{Statements of our main results}
\label{subsec:statements}
After these motivating examples, 
here we state our main results. From now on we mainly use the notations $c_4=E_4$ and $c_6=E_6$.

Our first result is the characterization of the mod-2 Witten genus $\varphi(\sT)$ of \Nequals{(0,1)} supersymmetric theory $\sT$ with gravitational anomaly $n=8k+2$. 
First we recall the concept of null bordism in the space of supersymmetric field theories.
Namely, when a theory can be deformed (in the sense discussed in \cite{Gaiotto:2019asa}) to break
supersymmetry spontaneously with a parametrically controlled vacuum energy, 
i.e.~so that the breaking scale is independent of the spatial size of the 2d spacetime,
we say that it is null-bordant in the space of supersymmetric field theories.

By abuse of notation, let $S^1$ denote the \Nequals{(0,1)}  supersymmetric sigma model on $S^1$ with periodic spin structure. 
Then our first result is:
\begin{claim}
\textbf{First statement:}\\
\PFirst
\end{claim}
We give the bulk of the derivation in Sec.~\ref{subsec_proof_first_phys},
which, roughly speaking, proceeds by imitating the analytic proof of Stong's proposition in footnote \ref{foot:analytic_stong} in the setting of supersymmetric quantum field theory.
Here we only prove the last part.
Integral modular forms of weight $4k+2$ have 
an integral basis spanned by $c_4^i c_6 \Delta^j$ where $i,j$ are  nonnegative integers satisfying $i+3j=k-1$.
The basis of weakly-holomorphic ones are given by allowing negative integers for $j$,
therefore the last sentence follows from the rest.

The physical argument in Sec.~\ref{subsec_proof_first_phys} has a parallel homotopy-theoretic counterparts, proving the following mathematical form of the First statement. 
Let $\sigm \colon \TMF \to \KO((q))$ be the morphism of ring spectra from $\TMF$ to $\KO((q))$.
We use the same symbol $\vphi$ for the homomorphism between their homotopy groups.
We then give a proof of the following mathematical statement in Sec.~\ref{subsec_proof_first_math}:

\begin{claim}
\textbf{Mathematical form of the First statement:}\\
\MFirst
\end{claim}

For our second statement, currently we only have mathematical derivation,
so we present the mathematical form first:\begin{claim}{}
\textbf{Second statement:}\\
\MSecond
\end{claim}
Our proof is not conceptual and rather is a simple corollary of the detailed structure of 
$\pi_\bullet\TMF$ described in \cite{BrunerRognes}.
In fact it allows us to characterize the image of $\vphi$ for $n=8k+1$ and $8k+2$ completely.

\begin{rem}
Actually, the detailed structure of $\pi_\bullet\TMF$ we use can also be used to prove the First statement above. 
However the proof we give for the First statement in Subsec.~\ref{subsec_proof_first_math} is essentially different,
which is geometric and does not rely on the detailed structure of $\pi_\bullet \TMF$.
\end{rem}

Assuming the Stolz-Teichner conjecture \cite{StolzTeichner1,StolzTeichner2} which states that \begin{equation}
\pi_n\TMF=\TMF_n(pt) = \pi_0(\left\{
\begin{array}{c}
\text{\Nequals{(0,1)} supersymmetric theories} \\
\text{with gravitational anomaly $n$}
\end{array}
\right\})
\end{equation}
and that $\sigm:\pi_n\TMF\to \pi_n\KO((q))$ applied on the left hand side 
can be identified 
with the ordinary and mod-2 Witten genera $\varphi$
applied on the right hand side,
we can restate the second result above as follows:
\begin{claim}{}
\textbf{Physics form of the second statement,  assuming the Stolz-Teichner conjecture:}\\
Let $\sT$ be a theory with gravitational anomaly $n=8k+2$. 
Then $\varphi(\sT)$ is a mod-2 reduction of a weakly-holomorphic integral modular form of weight $4k$.
More explicitly, it is given by a linear combination over $\physZtwo$ of $c_4^{k-3j}\Delta^j \equiv\Delta^j$,
where $j\le \lfloor \frac{k}{3}\rfloor$.
\end{claim}
Again, we only prove the last sentence of the claim here, relegating the rest to Sec.~\ref{yamashita}.
The derivation of the last sentence is entirely analogous to the case above;
the only difference is that the integral basis is now given by $c_4^i  \Delta^j$
with $i+3j=k$,
from which the last sentence follows.

Note that the first statement and the second statement have differing upper bounds for $j$ when $k=3\ell$.
For consistency, it should be that when $\varphi(\sT)=\Delta^\ell$,
$ \sT \times S^1$ should not be null-bordant within the space of supersymmetric quantum field theories.
In the language of $\TMF$, it should be that when $\vphi(x)=\eta^2 \Delta^\ell$ 
for $x\in \pi_{24\ell+2}\TMF$, we should have $\eta x\neq 0 \in \pi_{24\ell+3}\TMF$.
This in fact follows from the results of \cite{Bunke} (whose physics counterpart is discussed in \cite{Gaiotto:2019gef,Yonekura:2022reu}),
and we will also check that using the data provided in \cite{BrunerRognes} at the end of Sec.~\ref{subsec_proof_second}.

\section{First statement}
\label{yonekura}

\subsection{The derivation of the First statement}\label{subsec_proof_first_phys}
In this section, we derive our first statement:\begin{claim}
\textbf{First statement:}\\
\PFirst
\end{claim}

This statement can be regarded as a corollary of the results of \cite{Bunke,Gaiotto:2019gef,Yonekura:2022reu} as we will comment near the end of
this subsection.
We discuss it in the physics language, but the following argument also makes sense mathematically if we restrict to string manifolds.
Then at the end we will mention additional facts so that the statement applies to $\TMF$. 
A more direct and mathematically rigorous proof without relying on \cite{Bunke} will be given in Section~\ref{subsec_proof_first_math}.

Let us consider the theory $\sT \times S^1 \times \bR$. 
\textcolor{black}{The target space is noncompact in the $\bR$ direction, and the energy spectrum is not discrete. In such a case, the definition of the partition function by a formula like \eqref{foo} is not straightforward. 
Thus we do not use the naive formula \eqref{foo}.
Instead, we define a kind of Witten index 
following the definition of the Atiyah-Patodi-Singer (APS) index~\cite{Atiyah:1975jf}. It was discussed
in the context of supersymmetric field theories in \cite{Yonekura:2022reu}.\footnote{
If the partition function $Z(\tau,\overline{\tau})$ is defined in such a way that it is modular invariant, then it depends on $\overline{\tau}$ as well as 
$\tau$~\cite{Gaiotto:2019gef,Dabholkar:2020fde}. 
The definition discussed below roughly corresponds to the limit $\overline{\tau} \to -\sqrt{-1} \cdot \infty$ with $\tau$ fixed. This limit is not modular invariant even if $Z(\tau,\overline{\tau})$ is, and this fact will play an important role.
}
}

We consider states annihilated by the supercharge $Q$,\footnote{For string manifolds, $Q$ is a Dirac operator coupled to appropriate powers of
the tangent bundle as in \eqref{too}. } 
and count the number of such states as usual in the definition of the Witten index. 
However, if the target space is noncompact, 
we need to be careful about the normalizability of the wave function of states in the noncompact direction.\footnote{
For string manifolds, a ``wave function'' is a section of an appropriate bundle on which the Dirac operator acts.}
The APS index is defined as follows. For states with $(-1)^F=+1$, we count states whose wave function (or more precisely its absolute value) 
is bounded but not necessarily square normalizable.
On the other hand, for states with $(-1)^F=-1$, we only count states which are square normalizable.
Thus there is asymmetry in the treatment of states with $(-1)^F=+1$ and $(-1)^F=-1$. 

Using the above counting rule for the APS index, we can define $\varphi(\sT \times S^1 \times \bR)$ as in the case of compact theories,
which is not generally a modular form due to the existence of noncompact directions~\cite{Gaiotto:2019gef,Dabholkar:2020fde}.
Let us consider a quantity $b(\sT \times S^1)$ defined by
\beq\label{eq:invariant}
b(\sT \times S^1)= \frac12 \varphi(\sT \times S^1 \times \bR)  \mod \bZ((q)) + \mathbb{Q} \otimes \MF_{4k+2},
\eeq
which measures the deviation of $\varphi(\sT \times S^1\times \bR)$ from being modular.
About this quantity, we will show the following two facts which establish the First statement:
\begin{itemize}
\item We have $\varphi(\sT) =  \varphi(\sT \times S^1 \times \bR) \bmod 2\bZ((q))$ and hence
\beq\label{eq:bvarphi2rel}
b(\sT \times S^1)= \frac12 \varphi(\sT) \mod \bZ((q)) +  \mathbb{Q}  \otimes \MF_{4k+2}.
\eeq
\item $b(\sT \times S^1)$ is zero when $\sT \times S^1$ is null-bordant. 
\end{itemize}
The reason that these facts establish the First statement is as follows. When $\sT \times S^1$ is null-bordant,
$\varphi(\sT)$ is zero  modulo $ 2\bZ((q)) + ( \mathbb{Q} \otimes \MF_{4k+2})$. In other words,
there exists $\phi \in \bZ((q))$ such that
$\varphi(\sT) + 2\phi \in \mathbb{Q} \otimes \MF_{4k+2}$ where $\varphi(\sT)$ is lifted in an arbitrary way from $ \bZ_2((q))$ to $ \bZ((q))$.
Thus $\varphi(\sT) + 2\phi$ takes values in $ \bZ((q))\cap (\mathbb{Q}\otimes \MF_\bullet ) = \MF_\bullet$.
This is exactly the First statement.

The equality $\varphi(\sT) =  \varphi(\sT \times S^1 \times \bR) \bmod 2\bZ((q))$ can be directly seen. 
The wave functions of the states annihilated by $Q$
must be constant in the $S^1 \times \bR$ directions. These constant wave functions are bounded but not square normalizable.
By taking into account the definition of the APS index explained above,
the number of zero modes with $(-1)^F=+1$
is the same as the total number of zero modes of $\sT$, while there is no zero mode with $(-1)^F=-1$. Thus we get the desired equality. 

Next let us show that $b(\sT \times S^1)$ is zero when $\sT \times S^1$ is null-bordant. 
Suppose that $\sT \times S^1$ is null-bordant in the space of supersymmetric field theories.\footnote{
For string manifolds, the usual null bordism of a manifold implies the null bordism in the physical sense discussed here,
as shown in \cite{Gaiotto:2019asa}.
}
Then, we have a noncompact theory $\sY$ with the gravitational anomaly $8k+4$ that has a ``boundary at infinity'' of the form $\sT \times S^1$.
Here, by a ``boundary at infinity'', we mean that the noncompact direction of $\sY$ is of the form $\sT \times S^1 \times \bR_{<0}$.
For instance, let $x \in \bR$ be the standard coordinate of the sigma model target space $\bR$. We may couple this sigma model to $\sT \times S^1$
in such a way that the supersymmetry breaking vacuum energy of $\sT \times S^1$ increases rapidly as $x \to +\infty$,
and goes to zero as $x \to -\infty$. In this way, the region $x \to +\infty$ has an increasing potential energy and is effectively ``compact'' in the sense
of the energy spectrum. (For instance, a harmonic oscillator in quantum mechanics is compact in the sense that the energy spectrum is discrete, 
even though its
target space itself is noncompact.)
Thus the noncompactness comes only from the region $x \to -\infty$. This is an example of $\sY$. 
For more details see \cite{Gaiotto:2019asa,Gaiotto:2019gef,Johnson-Freyd:2020itv,Yonekura:2022reu}.

Now we glue two copies of $\sY$ to $\sT \times S^1 \times \bR$. We notice that $\sT \times S^1 \times \bR$ has two boundaries at infinity.
One boundary is $\sT \times S^1$, and the other boundary is the orientation reversal $\overline{\sT \times S^1}$.\relax 
\footnote{For the meaning of ``orientation reversal'' in general supersymmetric field theories, see \cite{Yonekura:2022reu}.}
However, since the orientation reversal is given by $\overline{\sT \times S^1} = \sT \times \overline{S^1} = \sT \times S^1$,
we just have two copies of $\sT \times S^1$ as the boundaries of $\sT \times S^1 \times \bR$. We glue them to the boundaries of two copies of $\sY$
and get a theory $\sW$ which is compact in the sense that the energy spectrum is discrete. 

The APS index has the gluing law
which is also true in general supersymmetric field theories \cite{Yonekura:2022reu}.
When we glue manifolds, the APS index of the manifold obtained by 
the gluing is the sum of the APS indices of the glued manifolds.
Thus we get\footnote{In the following formula, the APS index of $\sY$ needs to be defined such that
we count square normalizable wavefunctions for $(-1)^F=+1$ and bounded wavefunctions for $(-1)^F=-1$, which was a choice opposite to the one used before.
This point however does not play any important role in the following.}
\beq
\varphi(\sW) = \varphi(\sT \times S^1 \times \bR) + 2 \varphi(\sY).
\eeq
$\sW$ is a compact theory in the sense of the energy spectrum, and hence $\varphi(\sW)  \in \MF_{4k+2}$.
Moreover, when the gravitational anomaly is $8k+4$, the APS index is a multiple of 2 as in the case of the ordinary index.
Thus the number of states are always multiples of 2,
and hence we have $\varphi(\sY) \in \bZ((q))$. (Recall our inclusion of $\kappa=1/2$ when $n=8k+4$.) 
Thus, we see that $b(\sT \times S^1)$ defined in \eqref{eq:invariant}
vanishes when $\sT \times S^1$ is null-bordant. 
From this discussion, it is also clear that $\varphi(\sT)$ is the mod-2 reduction of $\varphi(\sW)$.
This concludes the physical derivation of the First statement.

Let us put the above discussion in the context of \cite{Bunke,Gaiotto:2019gef,Yonekura:2022reu}. The quantity $b(\sT \times S^1)$
is an invariant of supersymmetric field theories with the gravitational anomaly $8k+3$~\cite{Gaiotto:2019gef,Yonekura:2022reu}, 
applied to the special case where the theory is of the form $\sT \times S^1$. This invariant gives an obstruction to null bordism (i.e., 
continuous deformation to spontaneous supersymmetry breaking). We found above that $b(\sT \times S^1)$
is given by \eqref{eq:bvarphi2rel}. Therefore, if $\varphi(\sT)$ is not a mod-2 reduction of an element of $\MF_{4k+2}$,
the invariant is nonzero and $\sT \times S^1$ is not null-bordant. 


Finally, let us briefly sketch how to make the argument applicable to $\TMF$ without relying on physics. 
First we apply the above argument to the case where $\sT$ is a string manifold
and get the result about the case of string manifolds.
Then we use the results of Bunke and Naumann~\cite{Bunke},
where an invariant 
\beq
 b^{\rm tmf}: \pi_{4m-1} \mathrm{tmf} \to T_{2m}:=\frac{\mathbb{R} ((q))}{  \bZ((q)) + \mathbb{R} \otimes \MF_{2m}}
\eeq was defined.
Combined with $ \pi_\bullet M\mathrm{String} \to \pi_\bullet \tmf$, we also get an invariant 
$ b^{\rm top}:  \pi_{4m-1} M\mathrm{String} \to T_{2m}$. 
When $4m-1=8k+3$, $S^1 \times M$ is always spin null-bordant,
and this is the situation where all the invariants discussed in \cite{Bunke} 
(i.e. analytic $b^{\rm an}$, geometric $b^{\rm geom}$, and topological $b^{\rm top}$) coincide with each other~\cite[Lemma~3.4 and Theorem~4.2]{Bunke}.
Our description of the invariant $b(\sT \times S^1)$ above coincides with their analytic/geometric invariant by the APS index theorem,
where we use the obvious inclusion $\mathbb{Q} \subset \bR$.
Now by using the surjectivity of $\pi_\bullet M\mathrm{String} \to \pi_\bullet \tmf$ (see \cite[Theorem~6.25]{Hopkins2002} and \cite{Surjectivity}), 
we compute the values of
the invariant $b^{\rm tmf}$ for elements of the form $\eta  x \in \pi_{8k+3} \tmf$ where $\eta=[S^1] \in \pi_1 \tmf$ and $\forall x \in \pi_{8k+2}\tmf$,
and conclude that
\beq
b^{\rm tmf}(\eta x) = \frac12 \varphi(x) \mod \bZ((q)) + \mathbb{R} \otimes \MF_{4k+2}.
\eeq
Notice also that the APS index of a manifold of the form $M \times N$ where $M$ is closed and $N$ is compact and has a boundary
is just the product of the Atiyah-Singer index for $M$ and the APS index for $N$. This implies that
the above discussion is compatible with the multiplication by the periodicity element $\Delta^{24}$. 
Thus the result is valid for $\pi_\bullet\TMF = \pi_\bullet \tmf[\Delta^{-24}]$.

\subsection{The proof of the mathematical form of the First statement}\label{subsec_proof_first_math}
Recall the notation $\eta \in \pi_1 S$ introduced in Subsection~\ref{sec:intro_math}.
As remarked there, we also use the symbol $\eta$ to denote the generators in $\pi_1 M\spin \simeq \pi_1 M\mathrm{String} \simeq \pi_1 \TMF \simeq \pi_1 \KO \simeq \bZ/2$ by abuse of notation. 
We show the following statement.
\begin{claim}
\textbf{Mathematical form of the First statement:}\\
\MFirst
\end{claim}

{
We need some notations.
We work on the symmetric monoidal stable $\infty$-category of spectra where the sphere spectrum $\mathbb{S}$ is the monoidal unit. 
For a spectrum $E$ and an element $x \in \pi_N (\mathbb{S})$ represented by a map $x \colon \Sigma^{N}\mathbb{S} \to \mathbb{S}$ which we denote by the same symbol, 
we define the spectrum $E/x$ by the homotopy fiber sequence
\begin{align}\label{eq_R/x}
    \Sigma^N E \xrightarrow{x \cdot} E \to E/x, 
\end{align}
where the first map is the multiplication by $x$. 
We note that any morphism $a:E\to F$ is compatible with the structure of $\mathbb{S}$-modules,
and therefore induces a morphism $a:E/x\to F/x$ with the following morphism of fiber sequences, 
\beq
\vcenter{\xymatrix{
 \Sigma^{N} E \ar[r]^{~x \cdot} \ar[d]^{a } &  E \ar[r]^-{}\ar[d]^{a  } &  E/x \ar[d]^{a }  \\
 \Sigma^{N} F \ar[r]^{~x \cdot } & F \ar[r] & F/x & .
}}
\eeq
When $a : \Sigma^M E \to  E$ is given by a multiplication by $y  \in \pi_M (\mathbb{S})$, there is another type of induced morphism $y \cdot : \Sigma^M(E/x) \to E/(xy)$ with the following morphism of fiber sequences, 
\beq
\vcenter{\xymatrix{
 \Sigma^{N+M} E \ar[r]^{\quad x \cdot} \ar[d]^{\id } & \Sigma^M E \ar[r]^-{}\ar[d]^{y \cdot } &  \Sigma^M  (E/x) \ar[d]^{y \cdot }  \\
 \Sigma^{N+M} E \ar[r]^{\quad xy \cdot } & E \ar[r] & E/(xy) & .
}}
\eeq
By using the second type of the induced map, now consider the map induced by multiplication by $2 \in \pi_0 (\mathbb{S})$, 
\begin{align}
   2\cdot \colon E/\eta \to E/(2\eta) \simeq E \vee \Sigma^{2 }E.
\end{align}
We compose it with the map $\id_{E} \vee 0 \colon E \vee \Sigma^{2}E \to E$ and define
\begin{align}
    \multi(2) := (\id_{E} \vee 0) \circ (2 \cdot) \colon E/\eta \to E. 
\end{align}

\begin{proof}[Proof of the First statement]
We start from the following commutative diagram induced by $\sigm \colon \TMF \to \KO((q))$, 
\begin{align}
\vcenter{    \xymatrix@C=50pt{
     \pi_{8k+2}\TMF  \ar[d]^{\varphi} & \pi_{8k+4}(\TMF/\eta) \ar[l]_{\partial}\ar[r]^-{\multi(2)} \ar[d]^-{\sigm} & \pi_{8k+4}\TMF \ar[d]^-{\sigm}\\
   \pi_{8k+2}\KO((q))& \pi_{8k+4}(\KO((q)) /\eta)  \ar[l]_{\partial} \ar[r]^-{\multi(2)} & \pi_{8k+4} \KO((q)) 
    }}.\label{3.12}
\end{align}
To study the bottom row, we consider the homotopy commutative diagram whose rows are fiber sequences, 
\begin{equation}
\vcenter{\xymatrix{
 \Sigma \KO \ar[r]^{\eta \cdot} \ar[d]^{\id} & \KO \ar[r]^-{c}\ar[d]^{2 \cdot } & \KO/\eta \simeq \KU \ar[d]\ar[r]^-{\partial} & \Sigma^2\KO \ar[d]^{\id} \\
 \Sigma \KO \ar[r]^{2\eta \cdot =0} & \KO \ar[r] & \KO\vee \Sigma^2\KO \ar[r]& \Sigma^2\KO. 
}}
\end{equation}
Here we used Wood's theorem $\KO/\eta\simeq \KU$ under which the morphism $\KO\to \KO/\eta$ is given by the complexification $c:\KO\to \KU$, 
see e.g.~\cite[p.206]{Adams74} or \cite[1.2 Proposition]{snaith_1987}.
Taking the homotopy groups and using the fact that the complexification $c \colon \pi_{8k+4}\KO \simeq \bZ \to \pi_{8k+4}\KU \simeq \bZ$ maps a generator to twice a generator, we find that $\multi(2):\pi_{8k+4}(\KO/\eta) \simeq \bZ \to \pi_{8k+4}\KO\simeq \bZ$ is an isomorphism,
and that $\partial: \pi_{8k+4}(\KO/\eta)  \simeq \bZ \to \pi_{8k+2}\KO\simeq \bZ/2$ is the mod-2 reduction.
Therefore, the diagram \eqref{3.12} becomes 
\begin{equation}
\vcenter{    \xymatrix@C=50pt{
     \pi_{8k+2}\TMF  \ar[d]^{\varphi} & \pi_{8k+4}(\TMF/\eta) \ar[l]_{\partial}\ar[r]^-{\multi(2)} \ar[d]^-{\sigm} & \pi_{8k+4}\TMF \ar[d]^-{\sigm}\\
   \bZ/2((q))& \bZ((q))   \ar[l]_{\mod 2} \ar[r]^-{\simeq} & \bZ((q)) 
    }},
\end{equation}

Now, by the exact sequence $\pi_{8k+4}(\TMF/\eta) \xrightarrow[]{\partial} \pi_{8k+2} \TMF \xrightarrow{\eta \cdot} \pi_{8k+3}\TMF $, the subgroup $A_{8k+2} \subset \pi_{8k+2} \TMF$ is the image of $\partial \colon \pi_{8k+4}(\TMF/\eta) \to \pi_{8k+2} \TMF $. 
As the image of $\varphi$ on the rightmost vertical arrow is contained in $\MF_{4k+2}$,
the statement follows.
\end{proof}

\begin{rem}\label{rem_stong}
Here we show that the same technique can be used 
to prove the following:
    \begin{claim}
    \textbf{\cite[ Proposition on p.343]{StongTextbook}}\ 
Given a closed $(8k+2)$-dimensional spin manifold $M_{8k+2}$, choose a closed $(8k+4)$-dimensional spin manifold \[
W_{8k+4} := N \cup ([0, 1] \times S^1 \times M) \cup N\] constructed as in Subsec.~\ref{subsec_preliminary_mod2elliptic}. 
\textcolor{black}{Take any morphism of spectra $\pi:M\spin\to \KO$,
such as the ordinary Atiyah-Bott-Shapiro index.
Then  $\pi([M]) \in \pi_{8k+2}\KO \simeq \bZ/2$ is the mod two reduction of $\pi([W]) \in \pi_{8k+4}\KO \simeq \bZ$.}
\end{claim}
\begin{proof}
We can use the following commutative diagram, 
\begin{align}
    \vcenter{\xymatrix@C=50pt{
    \pi_{8k+2}M\spin \ar[d]^-{\pi} &\pi_{8k+4}(M\spin/\eta) \ar[d]^-{\pi} \ar[l]_-{\partial} \ar[r]^-{\multi(2)} &\pi_{8k+4}M\spin  \ar[d]^-{\pi} \\
    \pi_{8k+2}\KO \ar[d]^-{\simeq} &\pi_{8k+4}(\KO/\eta) \ar[l]_-{\partial} \ar[r]^-{\multi(2)}\ar[d]^-{\simeq} & \pi_{8k+4}\KO\ar[d]^-{\simeq} \\
    \bZ/2 & \bZ \ar[l]_-{\mod 2} \ar[r]^-{\simeq} & \bZ
    }},
\end{align}
which can be obtained exactly as before.

For general $i$, $\pi_{i+2}(M\spin / \eta)$ is the bordism group of pairs $(M_i, N_{i+2})$, where $M_i$ is an $i$-dimensional closed spin manifold and $N_{i+2}$ is a spin null bordism of $M \times S^1$, where $S^1$ is equipped with the nonbounding spin structure\footnote{A null-bordism of a pair $(M_i, N_{i+2})$ is given by a pair $(A_{i+1}, B_{i+3})$, where $A_{i+1}$ is an $(i+1)$-dimensional compact spin manifold equipped with an isomorphism $\partial A \simeq M$ and $B_{i+3}$ is a $\langle 2 \rangle$-manifold (a formulation of manifold with corners in \cite{Janich1968}) equipped with an isomorphism at each components of the boundary $\partial_0 B \simeq A \times S^1$, $\partial_1 B \simeq N$, which are compatible with the isomorphisms given at the corner $\partial_{01} B = \partial_0 B \cap \partial_1 B \simeq M \times S^1$. The justification of this geometric model can be given essentially in the same way as \cite[pp.25-26]{StongTextbook}}. 
Then the element $[M_{8k+2}, N_{8k+4}] \in \pi_{8k+4}(M\spin/\eta)$ maps to $[M_{8k+2}] \in \Omega_{8k+2}^\spin$ under $\partial$ and maps to $[W_{8k+4}] \in \Omega_{8k+4}^\spin$ under $\multi(2)$.
By the commutativity of the above diagram, 
the statement follows.
\end{proof}
\end{rem}
}

\section{Second statement}
\label{yamashita}

Here we show our second statement which is given in Subsection~\ref{sec:intro_math}:
\begin{claim}{}
\textbf{Second statement:}\\
\MSecond
\end{claim}

We provide two proofs. 
Our first proof is purely computational and does not shed much light on the mechanism behind it.
The second proof uses a lemma to reduce our Second Statement to our First Statement.
The lemma, unfortunately, is still proved purely computationally, 
but the authors expect that it can be given a more conceptual proof in the future.
This expectation is based on a physics consideration, which is also explained.

\subsection{A proof}
\label{subsec_proof_second}
\begin{proof}[The proof of the Second Statement]

We prove the statement obtained by replacing $\TMF$ by $\tmf$ and $\MF$ by $\mf$;
this suffices since we only have to invert the periodicity element $\Delta^{24}$ at the end.
Moreover we can localize at prime $2$,
so all the spectra are localized at $2$ in this proof. 

The proof for $\tmf$ is done by using the results of \cite[Chapter 9]{BrunerRognes},
where the structure of the ring $\pi_\bullet\tmf$ is described in detail.
We first note that $\sigm$ is multiplicative.
There is an element denoted by $B\in \pi_8\tmf$ 
which maps to $c_4\in \mf_4 \subset \pi_8\KO((q))$.
As $c_4$ is invertible in $\pi_\bullet\KO((q))$,
any $B$-power torsion in $\pi_\bullet\tmf$ maps to zero in $\pi_\bullet\KO((q))$.

Now, the structure of $\pi_\bullet\tmf$ modulo $B$-power torsion elements 
is given in Theorem 9.26 of \cite{BrunerRognes},
which says that the elements in degree $8k+1$ are in
\begin{equation}
\bZ/2[M,B]\{\eta,\eta_1,\eta B_2,\eta B_3,\eta_4,\eta B_5,\eta B_6,\eta B_7\} 
\end{equation} and those in degree $8k+2$ are in 
\begin{equation}
\bZ/2[M,B]\{\eta^2,\eta\eta_1, \eta_1^2,\eta^2B_3,\eta\eta_4,\eta_1\eta_4,\eta^2 B_6,\eta^2 B_7\},
\end{equation}
where \begin{itemize}
\item $M\in \pi_{192}\tmf$ is the periodicity element at prime $2$ such that $\varphi(M)=\Delta^8$,
\item  $B_k \in \pi_{24k+8}\tmf$,  such that $\varphi(B_k)=c_4 \Delta^k$,
\item $\eta=[S^1]\in \pi_1\tmf$ and $\eta_k\in \pi_{24k+1}\tmf$ are such that 
$\eta_k B=\eta B_k$.
\end{itemize}
The last relationship means that $\varphi(\eta_k)=\eta \Delta^k$,
implying that the image of $\varphi$ on $\pi_{8k+1}\tmf$ and $\pi_{8k+2}\tmf$
is contained in $\eta \mf_{4k}$ and $\eta^2\mf_{4k}$, respectively.
This was what we wanted to prove.

\end{proof}

\begin{rem}
In fact the information above completely determines the image.
The cokernel is spanned by $\eta\Delta^{8k+i}$ with $i=2,3,5,6,7$ 
and $\eta^2\Delta^{8k+j}$ with $j=3,6,7$.
Assuming the Stolz-Teichner conjecture,
this implies the absence of supersymmetric field theories whose mod-2 Witten genera are given by these elements.
\end{rem}

\begin{rem}
It is to be noted that the image of $\varphi$ is contained in $\eta^2 \times \MF_{4k}$ but not necessarily in $\eta^2  \times \varphi(\pi_{8k}\TMF)$.
For example, when $8k=24$, $\eta^2 \Delta$ is in the image, but $\Delta$ itself is not;
the smallest multiple of $\Delta$ which is in the image is $24\Delta$, as discussed in \cite{Hopkins2002}.
\end{rem}

\begin{rem}
We also mentioned at the end of Sec.~\ref{subsec:statements} that the compatibility of our first statement and the second statement requires that 
any element $x\in \pi_{24k+2}\tmf$ such that $\varphi(x)=\eta^2\Delta^k$ 
should satisfy $\eta x\neq 0 \in \pi_{24k+3}\tmf$.
This can be checked directly.
Note that  such $x$ are $\eta_i \eta_j$ where $i,j=0,1,4$
with the understanding that $\eta_0:=\eta$.
Now, the Proposition 9.34 (3) of \cite{BrunerRognes} states that $\eta \eta_i\eta_j 
\in \pi_{24(i+j)+3}\tmf$,
and is the unique nontrivial element of order 2 in that degree.
This is what we wanted to confirm.
\end{rem}

\subsection{An alternative proof}

The proof above was short but purely computational, and therefore did not provide much insight. 
Here we give an alternative proof, which relies on the following Lemma \ref{lem:e8}
whose proof is still computational in this paper, but could hopefully be done more conceptually, due to a physics reason we explain below.

\begin{lem}
\label{lem:e8}
Elements of $\pi_{8k+3}\TMF $ are annihilated by
the multiplication by the element $ e_8 \in \pi_{-16}\TMF$, 
where $e_8 \in \pi_{-16}\TMF$ is the element which is uniquely determined by the condition $\varphi(e_8)=c_4/\Delta$.
\end{lem}
\begin{proof}
Theorem 9.26 of \cite{BrunerRognes} says that all elements of $\pi_{8k+3}\TMF$ are $B$-power torsion.
Corollary 9.55 of the same reference says that any $B$-power torsion is annihilated by the multiplication by $B_7$.
As $\varphi(e_8)=c_4/\Delta$, $\varphi(B_7)=c_4\Delta^7$, and the periodicity element $M$ at prime 2 satisfies $\varphi(M)=\Delta^8$, our lemma follows.
\end{proof}

\begin{proof}[An alternative proof of the Second Statement]
Let $x\in \pi_{8k+2}\TMF$. 
From the lemma above, we have  $\eta e_8 x=0\in \pi_{8(k-2)+3}\TMF$.
Our First Statement is then applicable to $e_8 x\in \pi_{8(k-2)+2}\TMF$, 
showing that $\varphi(e_8 x)$ 
is the mod-2 reduction of an element $f$ in $\MF_{4(k-2)+2}$.
As $f$ is divisible by $c_6$,
and $c_4/c_6=1$ mod 2, 
$\varphi(e_8 x)$ is also the mod-2 reduction of an element $(c_4/c_6)f\in \MF_{4(k-2)}$.
Therefore $\varphi(x)$ is the mod-2 reduction of an element $(\Delta/c_6)f\in \MF_{4k}.$
\end{proof}

\begin{rem}
The proof of Lemma \ref{lem:e8} can be strengthened to prove the following:
\end{rem}
\begin{prop}
\label{prop:e8x}
The kernel of $\varphi:\pi_n\TMF\to \pi_n \KO((q)) $ is annihilated by
the multiplication by the element $e_8 \in \pi_{-16}\TMF$.
\end{prop}

\begin{proof}
In this proof the numbered theorems and corollaries refer to those in \cite{BrunerRognes}.
The map $\varphi$ is injective when rationalized,
and therefore we can discuss torsion elements for each $p$ separately.
For $p\ge 5$, there is no $p$-torsion in $\pi_\bullet\TMF$, and there is nothing to prove.
Let $\Gamma_B \pi_\bullet\TMF$ be the subgroup of $B$-power torsion elements of $\pi_\bullet \TMF$.
As $\varphi(B^k x)=\varphi(B)^k \varphi(x)$ and $\varphi(B)=c_4$ is invertible in $\pi_\bullet\KO((q))$, 
$\Gamma_B\pi_\bullet\TMF$ is in the kernel of $\varphi$.
Now, the structure of $\pi_\bullet\TMF/\Gamma_B\pi_\bullet\TMF$
is given in Theorem 9.26 (for $p=2$) and in Theorem 13.18 (for $p=3$),
and by inspection its generators all map to nonzero elements in $\pi_\bullet\KO((q))$.
Therefore any element in the kernel $\varphi$  is in $\Gamma_B\pi_\bullet\TMF$.
Now they are all annihilated by $B_7$ for $p=2$ by Corollary 9.55
and by $B_2$ for $p=3$ by Theorem 13.19.
As $e_8$ is $B_7/M$ for $p=2$ and is $B_2/H$ for $p=3$, where $M$ and $H$ are the periodicity elements for $p=2$ and $3$ respectively.
Therefore they are all annihilated by $e_8$.
\end{proof}

\begin{rem}
Let us explain why the authors would hope that a conceptual proof of Lemma~\ref{lem:e8} 
or its generalization Proposition~\ref{prop:e8x} might be found.
Here we utilize the twists of string bordisms and $\TMF$ by $H^4(-,\bZ)$,
which was described mathematically in \cite{ABG}
and whose physics interpretation was given in \cite[Appendix C]{Tachikawa:2021mby}.

Let $\nu \in \pi_3 \TMF$ be a generator of $ \pi_3 \TMF \simeq \bZ/24$.
In physics, the relation $e_8 \nu=0$ can be argued as follows.
On one hand, 
$\nu$ is the class of $S^3$ with the string structure coming from the trivialization of its tangent bundle coming from its Lie group structure.
This is not string null bordant, but has a spin null bordism, via $S^3=\partial B^4$.
In fact this null bordism can be promoted into a $BE_8$-twisted string null bordism,
where the twist is given by the generator $\tau$ of $H^4(BE_8,\bZ)=\bZ$,
by equipping $B^4$ with an $SU(2)\subset E_8$ bundle $P$ whose second Chern class $c_2(P)$ is 
supported away from the boundary, or in other words in the image of $H^4(B^4, S^3, \bZ)$, and such that $\int_{B^4}c_2(P)=1$.
This means that $\nu$ becomes null when it is sent via \begin{equation}
\iota_*: \TMF_3(pt)\to \TMF_{3+\tau}(BE_8).
\end{equation}

On the other hand,
the element $e_8\in \pi_{-16}\TMF\simeq \TMF^{16}(pt)$ is believed 
to have a realization in terms of a holomorphic vertex operator algebra $(e_8)_1$
under the Stolz-Teichner conjecture,
where the subscript $1$ specifies the level $\in H^4(BE_8,\bZ)=\bZ$.
For example, $\varphi(e_8) \in \MF_{-8}$ equals $\eta(q)^{-16}$ times the character of the unique integrable representation of the affine Lie algebra $(e_8)_1$;
this motivated our choice of the symbol $e_8$ for this element of $\pi_{-16}\TMF$.
In particular, this means that $e_8$ should lift to an element $\hat e_8$
of the $E_8$-equivariant version of $\TMF$ at level 1.
By going to the Borel equivariant version, this means that the element
$e_8\in \TMF^{16}(pt)$ should have a lift along \begin{equation}
\iota^*: \TMF^{16+\tau}(BE_8)  \to  \TMF^{16}(pt) 
\end{equation}
to an element $\hat e_8 \in \TMF^{16+\tau}(BE_8)$.

These two physics facts combined show that 
\begin{equation}
e_8 \nu=\iota^*(\hat e_8)\nu= \hat e_8\iota_*(\nu)=0.
\end{equation}
String theorists say that $\nu$ is the 3-sphere surrounding the heterotic NS5-brane,
that $e_8$ is the part of the worldsheet current algebra.
Then the fact that there exists the $E_8$-twisted string null bordism of $\nu$ leading to $e_8\nu=0$ means that 
the heterotic NS5-brane can be realized as an  $E_8$ instanton.
This property was first found in \cite{Witten:1996qz}.

The authors consider Proposition~\ref{prop:e8x}
as a natural mathematical generalization of this observation $e_8 \nu=0$,
by regarding $\nu$ as an example in the kernel of $\pi_\bullet\TMF\to \pi_\bullet\KO((q))$.
In this paper we only provided a computational proof for this proposition.
The authors have a way to lift not to $BE_8$ but to $K(\bZ,4)$ 
using the Anderson duality of $\TMF$, which will be explained elsewhere.\footnote{
A physical interpretation of this element $x$ of $\TMF^{16+\tau}(K(\bZ,4))$ is not yet clear.
Most naively, one might first hope that if $K(\bZ,4)$ has an $E_8$-bundle $P$ 
with the characteristic class $c_2(P)$ given by the fundamental class of $K(\bZ,4)$,
then we could pull back $\hat e_8 \in \TMF^{16+\tau}(BE_8)$ to $\TMF^{16+\tau}(K(\bZ,4))$
by the classifying map $f: K(\bZ,4) \to BE_8$,
which would be this element $x$.
However, this possibility 
is excluded as follows. We can consider a fibration $F \to BE_8 \to K(\bZ,4)$ where $g:BE_8 \to K(\bZ,4)$ 
is the map determined by the generator $\tau$ of $H^4(BE_8,\bZ)\simeq\bZ$,
and $F$ is the homotopy fiber such that $\pi_n(F)=0$ for $n<16$ and $\pi_{16}(F)\simeq \bZ$. If the map $f: K(\bZ,4) \to BE_8$ exists,
the composition $K(\bZ,4) \to BE_8 \to K(\bZ,4)$ is homotopic to the identity.
In this case, the Serre long exact sequence for the fibration gives a split short exact sequence
$0 \to H^{n}(K(\bZ,4) ) \to H^{n}(BE_8) \to H^{n}(F) \to 0$ (where $n<4+16-1$) with arbitrary coefficients. 
This contradicts with the fact that the transgression $H^{16}(F) \to H^{17}(K(\bZ,4))$ is nontrivial  for $\bZ/3$ coefficients \cite[p.259, Corollary]{Araki} and for $\bZ/2$ coefficients \cite[Lemma 1]{BE82}.
Therefore $f$ cannot exist.
It is possible to pull-back  $x\in\TMF^{16+\tau}(K(\bZ,4))$ to $\TMF^{16+\tau}(BE_8)$ by $g:BE_8\to K(\bZ,4)$,
but the resulting element $g^*(x)$ does not agree with $\hat e_8$ either.
This is because the image of $g^*(x)$ in $H^\bullet(BE_8,\bR)$
is a sum of powers of $\tau$ only, whereas the rationalization of $\hat e_8$ contains other characteristic classes in $H^\bullet(BE_8,\bR)$.
}

\end{rem}

\begin{rem}
\label{rem:phys}
Here is a physics interpretation of the crucial Lemma~\ref{lem:e8} and its generalization Proposition~\ref{prop:e8x}.
Suppose $M$ is a string manifold, which is the boundary of a spin manifold $N$, 
so that $M=\partial N$.
Suppose that the tangent bundle of $N$ is such that 
we can solve $dH=\frac{p_1}{2}(TN)-c_2(E_8)$ on $N$,
where $c_2$ is the instanton number of the $E_8$ bundle.
Assume furthermore that the $E_8$ bundle at the boundary $M$ is trivial.
A prototypical example is when $M=S^3$ with a unit $H$-flux.
Then we can simply take $N=B^4$, which we equip with the one-instanton configuration of $E_8$.
This is the realization of the heterotic NS5-brane as an $E_8$-instanton,
as originally found in \cite{Witten:1996qz}.

Now let us fiber the $E_8$ level 1 left-moving current algebra $(E_8)_1$ over $N$,
and perform the quantization over $N$.
This realizes a non-compact SQFT
whose boundary is the compact SQFT given by the product of $(E_8)_1$ and the sigma model on $M$,
meaning that $(E_8)_1M$ is null bordant in the space of SQFT.

What Proposition \ref{prop:e8x} says physically is that
this property holds more generally:
given an SQFT $x$ such that the ordinary or mod-2 elliptic genus of $x$ vanishes,
$(E_8)_1 x$ is null bordant in the space of SQFT.
It would be nice to have a physics-based derivation of Proposition \ref{prop:e8x} in general.

\end{rem}

\section*{Acknowledgments}
The authors thank A. Ishige for pointing out an error in an earlier manuscript,
and two anonymous referees for valuable comments which allowed the authors to significantly improve the quality of the article.
YT is supported in part  
by WPI Initiative, MEXT, Japan at Kavli IPMU, the University of Tokyo
and by JSPS KAKENHI Grant-in-Aid (Kiban-S), No.16H06335.
MY is supported by Grant-in-Aid for JSPS KAKENHI Grant Number 20K14307 and JST CREST program JPMJCR18T6.
KY is supported in part by JST FOREST Program (Grant Number JPMJFR2030, Japan), MEXT-JSPS Grant-in-Aid for Transformative Research Areas (A) ``Extreme Universe'' (No. 21H05188), and JSPS KAKENHI (17K14265).

\appendix

\section{Mod-2 index and $\KO_n(pt)$ in supersymmetric quantum mechanics}\label{sec:cliff}
For completeness, we explain more details about the mod-2 index and $\KO_n(pt)$ in the context of 
$\mathcal{N}=1$ supersymmetric quantum mechanics (SQM) by expanding on the idea in \cite{Gaiotto:2019gef} about spectator fermions. See also \cite{Gukov:2018iiq}. 
In this appendix $\bZ_n$ is the additive group of integers modulo $n$.

\paragraph{When the time reversal is non-anomalous:}

First let us consider the case that there is no anomaly in time reversal symmetry $T$
and the fermion parity $(-1)^F$. There are two possibilities for time reversal symmetry, $T^2=1$ or $T^2=(-1)^F$.
We are interested in the former $T^2=1$ since the (non-anomalous) CPT symmetry of 2d QFT gives such $T$ under dimensional reduction on a circle $S^1$ because of the following reasons. 
In general, let $\mathrm{Spin}(d,1)$ be the spin version of the Lorentz group in $(d+1)$-dimensions 
which is a double cover of $\mathrm{SO}(d,1)$.
If we compactify on $S^1$, there is a subgroup $\mathrm{O}(d-1,1) \subset \mathrm{SO}(d,1)$
such that orientation reversals in $\mathrm{O}(d-1,1) $ are realized as rotations in $\mathrm{SO}(d,1)$ involving $S^1$.
The corresponding double cover of $\mathrm{O}(d-1,1) $ is called $\mathrm{Pin}^-(d-1,1) \subset \mathrm{Spin}(d,1)$.
If we consider a ``$\pi$-rotation'' $R_{01}$ which flips the time direction and the $S^1$ direction,
we have $(R_{01})^2=1$ rather than $(-1)^F$, since the ``rotation'' involves the time direction. 
This corresponds to the time reversal $T$ in $\mathrm{Pin}^-(d-1,1)$.

In the absence of anomalies, $T$ and $(-1)^F$ are represented in the Hilbert space of SQM in the way dictated by 
$\mathrm{Pin}^-(0,1) \simeq \mathbb{Z}_2 \times \mathbb{Z}_2$.
Adding the supercharge $Q$,\footnote{Mathematically, $Q$ is a self-adjoint operator which we assume to have a discrete eigenvalue spectrum with finite dimensional eigenspaces for simplicity in this appendix.
$(-1)^F$ is a unitary operator, and $T$ is an anti-unitary operator. 
For instance, if we consider a sigma model whose target space is a closed spin manifold $M$, then $Q$ is a Dirac operator on $M$, $(-1)^F$ is a $\mathbb{Z}_2$-grading
in the Clifford module bundle on whose sections the Dirac operator acts,
and $T$ is roughly speaking the complex conjugation of sections. 
} 
the algebra is
\beq\label{eq:nonanomalous}
T^2=1, \quad ((-1)^F)^2=1, \quad (-1)^FT = T(-1)^F,  \quad TQ = QT, \quad (-1)^FQ = - Q(-1)^F.
\eeq
The transformation of $Q$ as $TQT^{-1}=Q$ is a choice of convention. If $TQT^{-1}=-Q$, we can modify $T$ to $(-1)^FT$
so that the new $T$ satisfies $TQT^{-1}=Q$.
The operator $T$ is an anti-linear operator acting on the Hilbert space with $T^2=1$, 
so the Hilbert space $\cH$ has a real structure,
\beq
\mathcal{H}=\mathcal{H}_{\mathbb{R}} \otimes_{\mathbb{R}} \mathbb{C},
\eeq
where $\mathcal{H}_{\mathbb{R}}$ is the subspace of $\mathcal{H}$ that is invariant under $T$, which  is a real vector space. 

\paragraph{The introduction of spectator fermions:}

When the theory has time reversal anomalies, let us introduce free spectator fermions so that the total anomaly is zero.
The time reversal anomalies are known to be classified by $\mathbb{Z}_8$, 
and a single Majorana fermion has the unit anomaly $1 \in \mathbb{Z}_8$~\cite{Fidkowski:2009dba,Kapustin:2014dxa,Witten:2015aba}.
In $\mathcal{N}{=}1$ SQM, a free Majorana fermion that is invariant under the supercharge $Q$ is possible.
(In the off-shell formulation of the supersymmetry, its superpartner is the auxiliary field without any degrees of freedom, and we set it to zero.)
When we are given an SQM theory, we add $n$ copies of such free Majorana fermions $\psi_i~(i=1,\cdots,n)$
so that the total theory has a non-anomalous time reversal symmetry. In the total theory,
$T, (-1)^F$ and $Q$ satisfy the above algebra \eqref{eq:nonanomalous}, and $\psi_i$ satisfy
\beq
\psi_i\psi_j+\psi_j\psi_i=2\delta_{ij}, \quad T\psi_i=\psi_i T, \quad (-1)^F\psi_i=-\psi_i(-1)^F, \quad Q\psi_i = - \psi_i Q.
\eeq
The first relation is the standard canonical anti-commutation relation of fermions. 
The second relation specifies the transformation of $\psi_i$
under the time reversal symmetry. A fermion with $T\psi=-\psi T$ is of course possible, but we have chosen our spectator fermions
to transform as above. Notice that the sign here is not a matter of convention but has a physical meaning since we have already fixed the convention
by the transformation of the supercharge $TQ = QT$ as discussed above.
The relation $T\psi_i=\psi_i T$ implies that $\psi_i$ is real in the sense that it 
preserves $\mathcal{H}_\mathbb{R}$.
Thus the above algebra implies that $\mathcal{H}_\mathbb{R}$
is a real representation of the $\mathbb{Z}_2$-graded Clifford algebra $\mathrm{Cliff}_{n}^{\bZ_2}$, 
where the $\mathbb{Z}_2$-grading is given by $(-1)^F$. Alternatively, we may define 
$\psi_0:=(-1)^F$
and regard $\psi_i~(i=0,\cdots, n)$ as the ungraded Clifford algebra with $n+1$ generators.

\paragraph{Hilbert space with and without spectator fermions:}
The Clifford algebra is convenient for general treatment of invariants of SQM. 
However, before studying that point, let us first discuss the relation to the algebras without any spectator fermions.
First we consider the case of even $n$. 
We define annihilation and creation operators as
\beq\label{eq:annihilation}
a_j= \frac{1}{2}(\psi_{2j-1}+ i \psi_{2j}), \quad a_j^\dagger= \frac{1}{2}(\psi_{2j-1}- i \psi_{2j})\qquad (j=1,\cdots, n/2)
\eeq
Then we may regard the subspace of the Hilbert space annihilated by $a_j$ as the Hilbert space
of the original theory without spectator fermions. Let $\mathcal{H}'= \{ v \in \mathcal{H}\,|\, a_j v=0~(j=1,\cdots,n/2)\}$ be this subspace.
However, this space is not invariant under $T$ since $T$ is anti-linear and hence $T a_j =  a_j^\dagger T$. 
Thus we need to modify $T$
if we want to obtain a time reversal symmetry of the original theory without $\psi_i$. This is done by defining 
\beq\label{eq:newT}
T' = T \cdot  \prod_{j=1}^{n/2} \psi_{2j-1} .
\eeq
One can check that $ T' a_j = \pm a_j T' $ and hence $\cH'$ is invariant under $T'$.
We may regard this as a time reversal symmetry of the original theory. Then we get the algebra
\beq\label{eq:even_alg}
n=2m : \quad T'^2=(-1)^{\frac12 m(m-1)},  \quad (-1)^F T' = (-1)^{m} T'(-1)^F,  \quad T'Q = (-1)^{m}  QT'.
\eeq
as well as the usual $((-1)^F)^2=1$ and $(-1)^FQ = - Q(-1)^F$. When $n=2 \mod 8$,
we reproduce \eqref{bosh}. When $n=4 \mod 8$, we obtain $T'^2=-1$ and hence
the Hilbert space $\mathcal{H}'$ has a pseudo-real structure (i.e., it is a quaternionic vector space) and 
in particular the dimension of the Hilbert space (or more precisely finite dimensional subspaces representing the algebra)
is always even. This is known as the Kramers degeneracy.  

For odd $n$, in addition to \eqref{eq:annihilation} for $j=1,\cdots, (n-1)/2$, we introduce additional
annihilation and creation operators
$b=\frac{1}{2} (a_n + i (-1)^F)$ and  $b^\dagger=\frac{1}{2} (a_n - i (-1)^F)$,
and consider the subspace $\cH'$ which is annihilated by $b$ as well as $a_j$. 
Notice that $(-1)^F$ is ill-defined on this subspace; this is a well-known anomaly of $(-1)^F$ for odd $n$.
We define
\beq\label{eq:newT2}
T' = T\cdot  \prod_{j=1}^{(n+1)/2} \psi_{2j-1} 
\eeq
and get
\beq\label{eq:odd_alg}
n=2m-1: \quad T'^2=(-1)^{\frac12 m(m-1)},   \quad T' Q = (-1)^{m}  Q T'.
\eeq

\paragraph{Clifford algebras and their modules:}
Let us come back to the Clifford algebras. If we consider a subspace of the Hilbert space on which $Q^2$
has an eigenvalue $Q^2=E >0$, we can define $\psi_{n+1}:=Q/\sqrt{E}$ on this subspace.
As a result, we get a Clifford algebra
$\psi_i \psi_j+\psi_j\psi_i=2\delta_{ij}$ for $i=0,\cdots, n+1$,
where $\psi_0=(-1)^F$ as defined above. This is the Clifford algebra $\mathrm{Cliff}_{n+1}^{\bZ_2}$.
Therefore, we are naturally led to consider the abelian group described in the next paragraph when we try to consider deformation
invariants of the Hilbert space of SQM. Exactly the same problem was studied in mathematics 
by Atiyah-Bott-Shapiro~\cite{Atiyah:1964zz}.\footnote{
In fact, a very closely related setup as in SQM has been discussed purely mathematically by Atiyah and Singer~\cite{Fredholm}. 
They considered the $\bZ_2$-graded Hilbert space $\cH_\mathbb{R}$ where the $\bZ_2$-grading is $(-1)^F$ in our notation, the action of the Clifford algebra
generated by $\psi_i$ (or skew-adjoint generators $J_i =(-1)^F \psi_i$ in their case),
and the space of self-adjoint Fredholm operators $\widetilde{Q}$ which are odd under the $\bZ_2$-grading and anticommute with 
the generators $\psi_i$ (or skew-adjoint operators $B=(-1)^F\widetilde{Q}$ in their case). 
The space of such $\widetilde{Q}$ is the classifying space for $\KO^{-n}$. 
More precisely, they considered the space of bounded operators with the norm topology, while  supercharges $Q$ of SQM are typically unbounded
and not norm-continuous under ``continuous'' deformations. 
This might be avoided by defining a bounded operator associated to $Q$ as e.g. $\widetilde{Q} = Q(Q^2+1)^{-1/2}$, and using the topology 
discussed in \cite[Definition~3.2]{Atiyah:2004jv}. (A standard assumption in ``compact'' quantum mechanics is that $e^{-\beta H}$, where
$H=Q^2$ in SQM and $\beta>0$, is a trace class operator  
which in particular implies that $1-\widetilde{Q}^2=(Q^2+1)^{-1}$ is a compact operator.)
}
 
The set of (isomorphism classes of) real representations 
of $\mathrm{Cliff}_{n}^{\bZ_2}$ forms a commutative monoid under the direct sum $\oplus$. 
We denote this commutative monoid as $M_{n}^{\bZ_2}$. There is a homomorphism $M_{n+1}^{\bZ_2} \to M_{n}^{\bZ_2}$
since we can regard $\mathrm{Cliff}_{n}^{\bZ_2} \subset \mathrm{Cliff}_{n+1}^{\bZ_2}$ by forgetting $\psi_{n+1} (=Q/\sqrt{E})$.
On the commutative monoid $M_{n}^{\bZ_2}$, we introduce the equivalence relation generated by $x \sim 0$ if $x \in M_{n}^{\bZ_2}$
actually comes from $ M_{n+1}^{\bZ_2}$.
It turns out that the commutative monoid becomes an abelian group after imposing this equivalence relation. 
We denote this abelian group as $\KO_n(pt)$ or $\KO^{-n}(pt)$,
\beq\label{eq:defKO}
\KO_n(pt):=\KO^{-n}(pt) := M_{n}^{\bZ_2}/M_{n+1}^{\bZ_2}.
\eeq
In the application to SQM, the subspace of $\mathcal{H}_\mathbb{R}$ with the eigenvalue $Q^2=0$
gives an element of $M_{n}^{\bZ_2}$, and other subspaces with $Q^2=E>0$ give elements of $M_{n+1}^{\bZ_2}$.
Thus, for a given SQM, we get an element of $\KO^{-n}(pt)$ which is invariant under continuous deformation of the theory;
even if a subspace with $E>0$ accidentally becomes $E=0$ during continuous deformation,
that subspace does not contribute to $\KO^{-n}(pt)$ since we have divided by $M_{n+1}^{\bZ_2}$.  

\def\C{C}
\def\SR{SR}
\def\PR{PR}

Irreducible representations of the Clifford algebra can be worked out explicitly. 
(See e.g. the appendices of \cite{Weinberg:2000cr,Polchinski:1998rr} for textbook treatments in physics.)
Irreducible complex representations of $\mathrm{Cliff}_{n}^{\bZ_2}$ 
have the structures summarized as follows.
We denote the existence of complex, strict-real, and pseudo-real structures as \C, \SR, and \PR, respectively. 
More precisely, \SR\ (resp.~\PR) means that
there exists an anti-linear operator $U$ with $U^2=1$ (resp. $U^2=-1$) that acts on the complex representation space
and commutes with $\mathrm{Cliff}_{n}^{\bZ_2}$. 
We also denote the complex dimensions of representations
by $\mathrm{dim}_\mathbb{C}$. Then $M_{n}^{\bZ_2}$ contains the following irreducible representations:
\begin{equation}\label{eq:rep_table}
\begin{array}{|c||c|c|c|c|c|c|c|c|c}
\hline
n &0+8k&1+8k&2+8k&3+8k \\ \hline
\text{reality} &\text{\SR}&\text{\SR}&\text{\C}&\text{\PR}\\ \hline
\text{dim}_\mathbb{C} &2^{4k}, 2^{4k}&2^{1+4k}&2^{1+4k}, 2^{1+4k}&2^{2+4k} \\ \hline
\hline
n &4+8k&5+8k&6+8k&7+8k \\ \hline
\text{reality}& \text{\PR}&\text{\PR}&\text{\C}&\text{\SR} \\ \hline
\text{dim}_\mathbb{C} &2^{2+4k}, 2^{2+4k}&2^{3+4k}&2^{3+4k}, 2^{3+4k}&2^{4+4k} \\ \hline
\end{array} \quad .
\end{equation}
One way to check the reality properties may be as follows. In this paragraph we regard the algebra as the ungraded one $\mathrm{Cliff}_{n+1}$
with $n+1$ generators $\psi_i~(i=0,1,\cdots, n)$.
Let us introduce $n+1$ additional Clifford generators $\eta_i$
such that $\eta_i \eta_j+\eta_j\eta_i=-2\delta_{i,j}$ and $\eta_i \psi_j+\psi_j\eta_i=0$
to get the Clifford algebra $\mathrm{Cliff}_{n+1, n+1}$ with signature $(n+1,n+1)$.
(Physically, $\eta_i$ are realized by fermions $\psi'_i$ satisfying $\psi'_i \psi'_j+\psi'_j \psi'_i=+2\delta_{i,j}$ and 
$T\psi'_i=- \psi'_i T$. We define $\eta_i = i \psi'_i$.)
By considering annihilation and creation operators $\frac{1}{2} (\psi_i \pm \eta_i)$, one can see that 
$\mathrm{Cliff}_{n+1, n+1}$ has a unique real irreducible representation of dimension $2^{n+1}$ which is constructed as a Fock space. 
Then we eliminate $\eta_i$ by considering the subspace annihilated by $\frac{1}{2}(\eta_{2j-1} + i \eta_{2j})$
similarly to the discussion following \eqref{eq:annihilation} to get a representation of $\mathrm{Cliff}_{n+1}$. 
Let $U$ be the anti-linear operator for the representation of $\mathrm{Cliff}_{n+1, n+1}$
with $U^2=1$, and let $\epsilon$ be 
$\epsilon = \prod (\psi_i \eta_i)$ which anti-commutes with $\psi_i$ and $\eta_i$ and satisfies $\epsilon^2=1$.
Then we define $U' = U \cdot \prod \epsilon \eta_{2j-1} $ similarly to \eqref{eq:newT} and \eqref{eq:newT2}. 
Notice that $ \epsilon \eta_{i}$ and hence $U'$ commute with $\psi_i$.
For even $n$, $U'$ satisfies the same algebra as \eqref{eq:even_alg} (without $Q$) by replacing $T' \to U'$ and $(-1)^F \to \epsilon \eta_0$.
We further restrict to eigenspaces of $\epsilon \eta_0$ to get two irreducible representations of $\mathrm{Cliff}_{n+1}$.
For odd $n$,  $U'$ satisfies the same algebra as \eqref{eq:odd_alg} (without $Q$) by replacing $T' \to U'$.
The value of $U'^2$, and the (anti-)commutation of $U'$ with $\epsilon \eta_0$ for even $n$, determine the reality properties
of the irreducible representations of $\mathrm{Cliff}_{n+1}$.

Let us also list some basic properties:
\begin{itemize}
\item For even $n$, we have two different irreducible representations of dimension $2^{n/2}$ 
on which $\prod_{i=0}^{n}\psi_i$ acts as $+ i^{n/2} $ and $- i^{n/2} $, respectively. (Notice that $\prod_{i=0}^{n}\psi_i=(-1)^F \prod_{i=1}^{n}\psi_i$ 
is a center of the algebra $\mathrm{Cliff}_{n}^{\bZ_2}$ for even $n$ with $(\prod_{i=0}^{n}\psi_i )^2=(-1)^{n/2}$.)
\item  For odd $n$, there is a unique irreducible representation of dimension $2^{(n+1)/2}$. 
It becomes the direct sum of the two different representations 
for $n-1$ under $M_{n}^{\bZ_2} \to M_{n-1}^{\bZ_2}$.
\item For $n=2 \mod 4$, the two representations are complex conjugates of each other. 
\end{itemize}

When we have a representation with C or PR structure, we obtain the corresponding strict-real representation 
by regarding it as a real vector space with the real dimension $\mathrm{dim}_\mathbb{R} = 2\mathrm{dim}_\mathbb{C}  $. 
For a representation with SR structure, we get a real representation with $\dim_\mathbb{R} = \dim_\mathbb{C}$.

\paragraph{Computations of $\KO^{-n}_G(pt)$:}

We can compute the groups $\KO^{-n}(pt)$ by using \eqref{eq:rep_table}.
For instance, let $V_1$ and $V_2$ be the two real irreducible representations of $\mathrm{Cliff}_{8k}^{\bZ_2}$.
Elements of $M_{8k}^{\bZ_2}$ are of the form $n_1V_1 \oplus n_2 V_2 $ for 
$(n_1, n_2) \in \mathbb{Z}_{\geq 0} \times \mathbb{Z}_{\geq 0}$. Elements of the image $M_{8k+1}^{\bZ_2} \to M_{8k}^{\bZ_2}$
are $n(V_1 \oplus V_2)$ for $n \in \mathbb{Z}_{\geq 0}$. We impose the equivalence relation
$(n_1, n_2) \sim (n_1+n, n_2+n)$ and get $\KO^{-8k}(pt) \simeq \mathbb{Z}$. The integer 
$n_1-n_2 \in \mathbb{Z} \simeq \KO^{-8k}(pt)$ is the usual Witten index.

Other groups can be computed similarly. 
However, it is not difficult to slightly generalize the problem as follows. 
Suppose that the theory has an internal global symmetry $G$ (which is compact)
such that it does not have anomalies except for pure time-reversal anomalies.
In this situation, we regard $M_{n}^{\bZ_2}$ as the commutative monoid of real representations of not only $\mathrm{Cliff}_{n}^{\bZ_2}$
but also $G$, such that the actions of $\mathrm{Cliff}_n^{\bZ_2}$ and $G$ commute.
 At the level of complex representations, we can construct irreducible representations by just taking tensor products
of irreducible representations of $\mathrm{Cliff}_{n}^{\bZ_2}$ and $G$. Notice that when we take the tensor product
of two representations each of which has PR structure, then we get a representation with SR structure. 

Let $R_\mathbb{C}(G)$, $R_\mathbb{R}(G)$ and $R_\mathbb{H}(G)$ be the abelian groups 
of complex, strict-real and pseudo-real representations of $G$, respectively. 
(More precisely, elements of these groups are virtual representations allowing negative coefficients.)
Also let $\KO^{-n}_G(pt)$ be
the abelian group defined as in \eqref{eq:defKO} in the presence of $G$. 
Then by using \eqref{eq:rep_table} and the properties mentioned above,
we get the following results for $\KO^{-n}_G(pt)$ which are known in mathematics (see e.g. \cite{AtiyahSegal} for the case including $G$):
\beq\label{eq:KOG}
\begin{array}{|c||c|c|c|c|c|c|c|c|c}
\hline
\! n~\mathrm{mod}~8 \! &0&1&2&3\\ \hline
\! \KO^{-n}_G(pt) \! &\! R_\mathbb{R}(G) \!&\! R_\mathbb{R}(G)/ R_\mathbb{C}(G) \!&\! R_\mathbb{C}(G)/R_\mathbb{H}(G) \!& 0
 \\ \hline
\! \KO^{-n}(pt) \!&\mathbb{Z}&\mathbb{Z}_2&\mathbb{Z}_2&0\\ \hline
\hline
\! n~\mathrm{mod}~8 \! &4&5&6&7 \\ \hline
\! \KO^{-n}_G(pt) \! 
&\! R_\mathbb{H}(G) \!&\! R_\mathbb{H}(G)/R_\mathbb{C}(G) \!&\! R_\mathbb{C}(G)/R_\mathbb{R}(G) \!& 0 \\ \hline
\! \KO^{-n}(pt) \! & \mathbb{Z}&0&0&0 \\ \hline
\end{array} \quad .
\eeq
Here the meaning of $R_\mathbb{R}(G)/ R_\mathbb{C}(G)$ and $R_\mathbb{H}(G)/R_\mathbb{C}(G)$
is that for $V \in R_\mathbb{C}(G)$ we take $V \oplus \overline{V}$ regarded as a (strict or pseudo) real representation
and impose $V \oplus \overline{V} \sim 0$ in $R_\mathbb{R}(G)$ and $R_\mathbb{H}(G)$, respectively. 
The meaning of $R_\mathbb{C}(G)/R_\mathbb{H}(G)$
and $R_\mathbb{C}(G)/R_\mathbb{R}(G)$ is that for $V \in R_\mathbb{R}(G)$ or $R_\mathbb{H}(G)$ we complexify $V$
and then impose $V \sim 0$ in $R_\mathbb{C}(G)$. For convenience, we have also explicitly given $\KO^{-n}(pt)$
which is obtained from more general $\KO^{-n}_G(pt)$ by taking $G=1$.
The mod-2 index of our interest in this paper is the group $\KO^{-n}(pt) \simeq \mathbb{Z}_2$ for $n =1,2 \mod 8$.

\paragraph{The case of 2d supersymmetric quantum field theories:}
Finally, let us briefly comment on the case of 2d $\mathcal{N}{=}(0,1)$ SQFT~\cite{Gaiotto:2019gef}.
The situation is basically the same as in the case of SQM as far as the mod-2 index is concerned.
We introduce spectator left-moving Majorana-Weyl fermions which are invariant under the supercharge $Q$
so that the total gravitational anomaly (or at least the time-reversal anomaly after dimensional reduction to SQM) is zero.
If we consider the theory on $S^1$ with the periodic (Ramond) spin structure, we get zero modes from the Majorana-Weyl fermions.
These zero modes play the role of $\psi_i$ in the above discussion. The contribution of the excited modes of the Majorana-Weyl fermions
gives the Dedekind $\eta$-function $\eta(\tau)$ which appear in the main text of this paper. 
In this way we get the invariant which takes values in $ \KO_G^{-n}((q))(pt)$. For $n=1,2$ mod 8 and $G=1$,
this invariant is the mod-2 Witten genus studied in this paper.

\def\arxivfont{\rm}
\bibliographystyle{ytamsalpha}

\baselineskip=.98\baselineskip

\bibliography{ref}

\end{document}